\documentclass[aps,prx,twocolumn,groupedaddress,nofootinbib]{revtex4-2}

\usepackage{graphicx}
\usepackage{algorithm}
\usepackage{algpseudocode}
\usepackage{amssymb,amsmath}
\usepackage{upgreek}
\usepackage[colorlinks=true, linkcolor=red, citecolor=blue]{hyperref}
\usepackage{cleveref}
\usepackage{amsthm}
\usepackage{bbm}

\DeclareMathOperator{\wg}{wg}
\DeclareMathOperator*{\E}{\mathbb{E}}
\DeclareMathOperator{\weight}{weight}
\DeclareMathOperator{\supp}{supp}
\DeclareMathOperator{\flips}{flips}

\newcommand{\nc}{\newcommand}
\newcommand{\ot}{\otimes}

\nc{\wt}[1]{\widetilde{#1}}
\nc{\polylog}{\operatorname{polylog}}
\nc{\expval}[2]{\< #2 | #1  | #2 \>}
\newcommand{\bra}[1]{\langle #1|}
\newcommand{\ket}[1]{|#1\rangle}
\newcommand{\braket}[2]{\langle #1|#2\rangle}
\def\t{\theta}
\makeatletter
\newcommand*{\defeq}{\mathrel{\rlap{%
                     \raisebox{0.3ex}{$\m@th\cdot$}}%
                     \raisebox{-0.3ex}{$\m@th\cdot$}}%
                     =}
\makeatother
\newcommand{\tr}{{\rm tr}}
\newcommand {\br} [1] {\ensuremath{ \left( #1 \right) }}
\newcommand {\cbr} [1] {\ensuremath{ \left\lbrace #1 \right\rbrace }}
\newcommand {\sbr} [1] {\ensuremath{ \left[ #1 \right] }}
\newcommand {\abr} [1] {\ensuremath{ \left| #1 \right| }}
\nc{\dg}{\dagger}
\nc{\be}{\begin{equation}}
\nc{\ee}{\end{equation}}
\newcommand{\bbR}{\mathbb{R}}
\newcommand{\bbC}{\mathbb{C}}

\nc{\pd}{\partial}
\nc{\proj}[1]{|#1\rangle\!\langle #1 |} 

\newtheorem{theorem}{Theorem}
 
\newtheorem{lemma}{Lemma}
\newtheorem{corollary}{Corollary}

\def\bx{\mathbf{x}}

\def\bn{\mathbf{n}}
\def\bze{\mathbf{0}}
\def\but{\mathbf{\uptheta}}
\def\>{\rangle}
\def\<{\langle}
\def\cM{{\cal M}}
\def\cS{{\cal S}}
\def\cC{{\cal C}}

\def\g{\gamma}
\def\d{\delta}
\def\ep{\varepsilon}
\def\eps{\varepsilon}

\def\t{\theta}

\def\p{\pi}

\def\s{\sigma}
\def\ta{\tau}

\def\G{\Gamma}

\def\vg{\vec{\gamma}}
\def\vG{\vec{\Gamma}}
\def\ba#1\ea{\begin{align}#1\end{align}}

\renewcommand{\mtt}[1]{\mathtt{#1}}

\begin{document}


\title{Quantifying the barren plateau phenomenon for a model of unstructured variational ans\"{a}tze}


\author{John C. Napp}
\email{john.napp@gmail.com}
\affiliation{Center for Theoretical Physics, Massachusetts Institute of Technology, Cambridge, MA 02139, USA}


\date{\today}

\begin{abstract}
Quantifying the flatness of the objective-function landscape associated with unstructured parameterized quantum circuits is important for understanding the performance of variational algorithms utilizing a ``hardware-efficient ansatz'', particularly for ensuring that a prohibitively flat landscape---a so-called ``barren plateau''---is avoided. For a model of such ans\"{a}tze, we relate the typical landscape flatness to a certain family of random walks, enabling us to derive a Monte Carlo algorithm for efficiently, classically estimating the landscape flatness for any architecture. The statistical picture additionally allows us to prove new analytic bounds on the barren plateau phenomenon, and more generally provides novel insights into the phenomenon's dependence on the ansatz depth, architecture, qudit dimension, and Hamiltonian combinatorial and spatial locality. Our analysis utilizes techniques originally developed by Dalzell et al. \cite{dalzell2020random} to study anti-concentration in random circuits.  
\end{abstract}


\maketitle

\section{Introduction}
A leading candidate class of algorithms for obtaining a quantum speedup for a practical problem in the near term is that of variational hybrid quantum-classical algorithms (see \cite{cerezo2020variational} for a review). In this setting, one assumes access to a quantum device capable of implementing some parameterized family of quantum circuits $U(\but)$ for $\but \in \bbR^p$, and the goal is to minimize an objective function $f$ of the form $f(\but) = \expval{U(\but)^\dag H U(\but)}{0^n}$ over the set of feasible parameters. While perhaps the most well-known algorithms of this form are the Variational Quantum Eigensolver \cite{peruzzo2014variational} (for estimating the ground state energy associated with some physical Hamiltonian) and the Quantum Approximate Optimization Algorithm \cite{farhi2014quantum} (for solving combinatorial optimization problems), it is possible to encode a wide variety of disparate computational problems in this form by making an appropriate choice of \emph{objective observable} $H$. With $H$ and the parameterized circuit $U(\but)$ having been set, one generally proceeds with the minimization of the objective function $f$ via an interaction between the quantum device and a classical controller; the quantum device is used to estimate the objective function or its derivatives at any point in the parameter space, while the classical controller is used to perform an ``outer loop'' stochastic optimization of $f$ over the parameter space.

This framework is appealing from the perspective of near-term quantum applications due to its great flexibility; in addition to the flexibility with respect to the computational problem encoded in $H$, there is also an enormous flexibility with respect to the available quantum hardware. While the parameterized circuit $U(\but)$ may be chosen in some highly-structured, theoretically-motivated way for the problem at hand---which could necessitate quantum resources beyond what are available in the near term---another strategy is simply to choose $U(\but)$ to be naturally compatible with the available hardware. An ansatz chosen in this manner is deemed a \emph{hardware-efficient ansatz} (HEA). As an HEA ignores the structure inherent to the problem, it typically ``looks random'' from the perspective of $H$. The success of a variational algorithm depends largely on (1) how well the ansatz $U(\but)$ can express the ground state of $H$, and (2) the geometric properties of the objective function landscape influencing how easily optimization may be performed. The focus of this work is on the latter point, particularly in relation to the HEA.

Despite the allure of the HEA for providing a possible route to practical quantum speedups in the near term, it is well-known that such highly-unstructured, random-looking ans\"{a}tze suffer from a drawback known as the \emph{barren plateau phenomenon} \cite{McClean_2018}, which is the tendency of the objective function landscape to look extremely flat almost everywhere.

The flatness of the objective function landscape is indeed a key property affecting the performance of a variational quantum algorithm. If the landscape is extremely flat, then intuitively the classical ``outer loop'' optimization will have trouble finding a good local minimum. Furthermore, gradient-based optimization approaches would need to take an enormous number of measurements at a typical point in parameter space to estimate $\nabla f$ at that point with small relative error. On the other hand, if the landscape is sufficiently ``bumpy'', then it is typically possible to obtain a good estimate of $\nabla f$ from a reasonable number of measurements, and an optimization approach like stochastic gradient descent can be used to descend to a local minimum. The importance of the barren plateau phenomenon is also unsurprising given its close relation to the famous ``vanishing gradient problem'' \cite{hochreiter_2001} encountered in the training of (classical) deep neural networks. Due to its direct relevance to the performance of optimization algorithms, the barren plateau phenomenon has been studied directly or played a role in a large number of recent numerical and analytical works \cite{McClean_2018, Grant_2019, Cerezo_2021_cost, Holmes_2021, Uvarov_2021, Cerezo_2021_higher, Abbas_2021, Volkoff_2021, skolik2021layerwise, Zhao_2021, sharma2020trainability, wang2021noiseinduced, marrero2021entanglement, zhang2020trainability, arrasmith2020effect, pesah2020absence, patti2020entanglement, holmes2021connecting, arrasmith2021equivalence, haug2021optimal, larocca2021diagnosing, kim2021entanglement, anschuetz2022critical, sack2022avoiding, rad2022surviving}.

But despite the importance of the phenomenon and the plethora of recent work, there remain practical and theoretical holes in our understanding. For instance, on the practical side, for a given circuit architecture and observable $H$ there is no known general method for efficiently estimating the flatness of the associated objective function landscape. On the theoretical side, analytic results tend to either apply only to narrow special cases of architectures, or have general applicability but be quite loose. For example, the analytic bounds in the original paper \cite{McClean_2018} apply to the practically unrealistic setting of ans\"{a}tze which form exact unitary 2-designs. This was improved in \cite{holmes2021connecting}, which derived upper bounds on the gradient for approximate 2-designs; however, there was no lower bound, and in general ans\"{a}tze need not be approximate 2-designs to experience barren plateaus. In \cite{Cerezo_2021_cost}, the authors proved some upper and lower bounds on the typical magnitude of the gradient for a model of a one-dimensional, $O(\polylog n)$-depth HEA where the objective observable $H$ is either spatially local and few-body or is a tensor product of non-trivial projectors on each qubit. The techniques of the present work allow us to essentially generalize some of their results to arbitrary depth and more general Hamiltonians. In \cite{Uvarov_2021} the authors proved a very general lower bound on the variance of the gradient, but one implication of the present work is that this bound can be very loose. More recently, \cite{Zhao_2021} employed the ZX-calculus to derive bounds for four special cases of circuit architecture. Notably, it was not well-understood in general how the flatness of the landscape scales with the variational circuit depth, qudit dimension, or the locality (combinatorial or spatial) of $H$.

In this work we help ameliorate some of these gaps in knowledge. We first introduce a model of highly-unstructured HEAs which our results are derived with respect to, which is fundamentally similar to the models used in \cite{Cerezo_2021_cost, Uvarov_2021} but allows for generalization beyond their setup (e.g. to more general types of parameterized gates and to qudits with dimension greater than two). Within this model, for arbitrary architectures we give an efficient Monte Carlo algorithm for estimating the typical magnitude of $\nabla f$, and additionally derive general upper and lower analytic bounds on $\nabla f$. Stronger analytic bounds are obtained for 1D architecture, where we effectively generalize some of the results of \cite{Cerezo_2021_cost}. Perhaps the most important novel theoretical implications of our results are that they show or suggest that, typically, (1) the gradient decays exponentially in the circuit depth; (2) the gradient decays exponentially in the Hamiltonian locality; and (3) the gradient decays polynomially in the local qudit dimension. As elaborated upon in the discussion section, (2) in turn suggests that even for architectures in which barren plateaus are avoided, gradient descent may tend to find local optima which neglect higher-order terms of $H$; additionally, it implies that global observables are \emph{always} associated with barren plateaus in this model, regardless of architecture (generalizing a result of \cite{Cerezo_2021_cost} for narrower classes of architectures and observables). Point (1) implies that the landscape flatness does not generally saturate at the 2-design depth. Our work also resolves an open question posed in \cite{Cerezo_2021_cost} on the relative importance of the combinatorial versus spatial locality of $H$ in determining the flatness of the landscape; combinatorial locality is generally the dominant factor, but the spatial locality structure also contributes as well in a way that is easy to intuit after deriving the Monte Carlo algorithm for barren plateaus below. 

\section{Setup and notation}
We now describe the HEA model. The parameterized circuit $U(\but)$ is assumed to act on $n$ qudits of local dimension $q$. The starting state is assumed to be $\ket{0}^{\ot n}$. $U(\but)$ is assumed to consist of two types of gates, which we call \emph{entangling gates} and \emph{parameterized gates}. Entangling gates act non-trivially on two (possibly non-adjacent) sites and are randomly chosen according to any measure on U$(q^2)$ that forms a 2-design. A parameterized gate is of the form $W_l e^{-i A \theta} W_r$ where $A$ is Hermitian and acts non-trivially on at most two sites, $\theta\in \bbR$ is the parameter, and $W_l$ and $W_r$ are arbitrary fixed gates acting on the same sites as $A$. We assume that $e^{-i A \theta}$ is periodic in $\t$ with period $2\pi$ (re-scaling $A$ if necessary), and take $[0,2\pi]$ as the feasible set for each parameter. The model additionally obeys the following constraints. First, each qudit is acted on by an entangling gate at least once. Second, each entangling gate acting on qudits $i$ and $j$ may be preceded and succeeded by an arbitrary number of parameterized gates acting on one or both of these sites, but parameterized gates which cannot be placed in this way are not allowed. (There is no constraint on the location of entangling gates.) Third, for our \emph{lower bounds} on the gradient, we additionally assume that the final gates to act on any given qudit $i$ is the parameterized sequence $\Sigma_1^{\alpha_i q/ 2\pi} \Sigma_3^{\beta_i q / 2\pi}$ where $\alpha_i$ and $\beta_i$ are parameters, and $\Sigma_1$ and $\Sigma_3$ are the so-called shift and clock matrices, respectively, defined by $\bra{k}\Sigma_1\ket{l} \defeq \delta_{k-1,l}$ and  $\bra{k}\Sigma_3\ket{l} \defeq e^{2\pi i k/q}\delta_{k,l}$.\footnote{These parameterized gates may be expressed in the required form and have period $2\pi$. Note that this final assumption is not necessary for the upper bounds, and additionally, for the case of qubits ($q=2$), we may replace $\Sigma_1^{\alpha_i q/ 2\pi} \Sigma_3^{\beta_i q / 2\pi}$ with the sequence of Pauli rotations $e^{i \alpha_i X/2}e^{i \beta_i Z/2}$ and the results remain unchanged.}

We additionally define a few more pieces of notation. For any valid variational circuit $U(\but)$, we define the associated random circuit $\wt{U}$ to be the same but with all parameterized gates removed. The variable $m$ is used to denote the number of gates in $\wt{U}$. The variable $d$ denotes the depth (i.e. number of layers of parallel gates) of $\wt{U}$. In analogy with \cite{dalzell2020random}, variable $r$ denotes the \emph{regular connectivity} of $\wt{U}$, defined to be the maximum number of layers of parallel gates that must be applied before some gate acts between an arbitrary proper subset of qudits and its complement.  Define $p$ to be the number of parameters, so the feasible set in parameter space is $[0,2\pi]^p$.  The model is illustrated by example in \Cref{fig:circuit}. We also define $[p] \defeq \{0, \dots, p-1\}$, and $\bze$ denotes an all-zeros vector.

\begin{figure}
    \includegraphics[width=0.8\columnwidth]{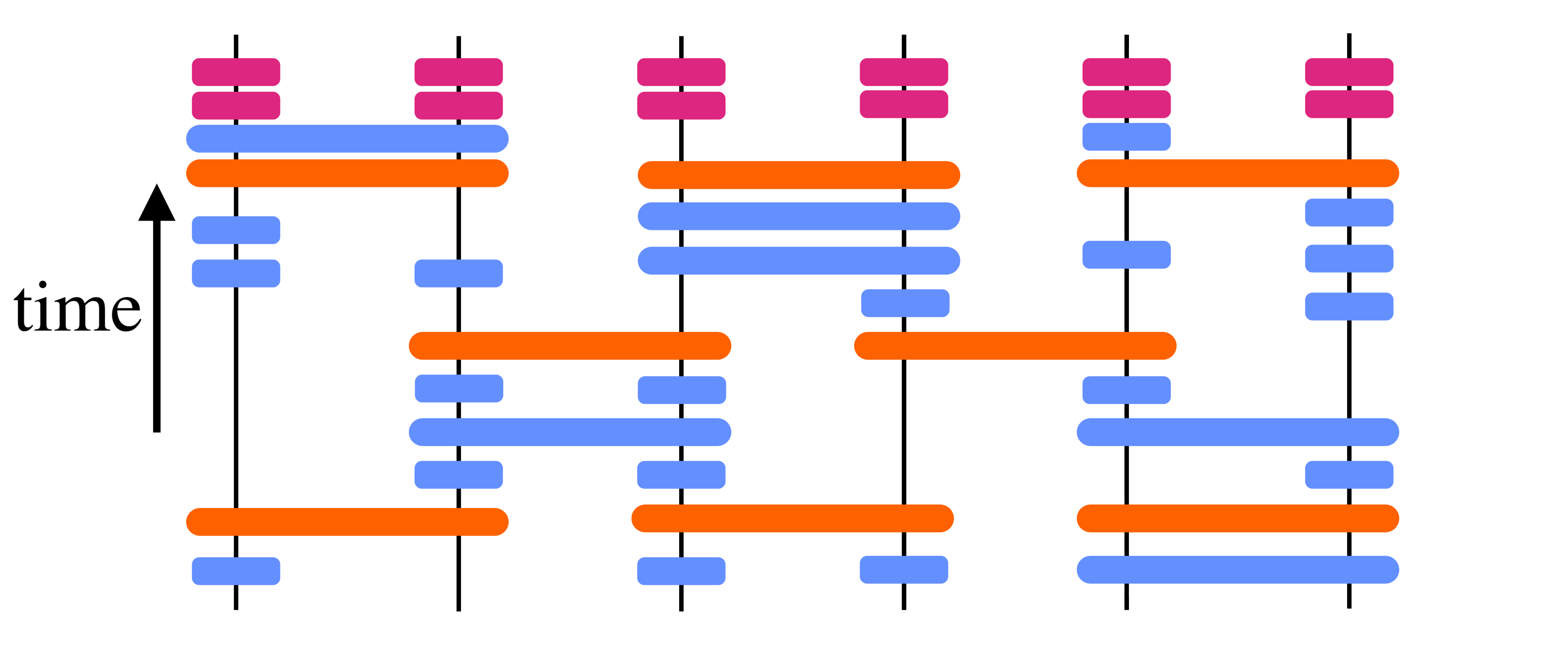}
    \caption{Example of an $n=6$ qudit architecture compatible with the model. Time flows from bottom to top. Orange represents entangling gates, blue represents arbitrary parameterized gates, and pink represents the final parameterized gates to act as specified in the main text. In this example, there are $m=8$ entangling gates, there are $p=36$ variational parameters, the entangling gate depth is $d=3$, and the entangling gate regular connectivity is $r=2$. While a 1D architecture is illustrated for simplicity, no such constraint is required for the Monte Carlo algorithm or general analytic bounds, although stronger bounds are obtained for the 1D setting.}
    \label{fig:circuit}
\end{figure}

Let $\{M_0, \dots, M_{q^2-1}\}$ be any set of matrices which form an orthogonal basis for $\bbC^{q \times q}$ w.r.t. the Hilbert--Schmidt inner product, satisfying $M_0 = I$ (the identity matrix) and $\tr (M_i^\dg M_j) = q \d_{i,j}$. In the case of qubits (i.e. $q=2$), for example, these may be the Pauli matrices. For higher $q$, $\{\Sigma_1^i \Sigma_3^j\}_{(i,j)\in [q]^2}$ would do. Now, $H$ may be decomposed as $H = \sum_{\bx} c_{\bx} M_{\bx}$, where we have introduced the notation
\be
M_{\bx} \defeq M_{\bx_1} \ot \dots \ot M_{\bx_n},
\ee
where $\bx \defeq (\bx_1, \dots, \bx_n)$ and  $\bx_{i} \in [q^2]$. Note that $c_\bx = q^{-n} \tr(M_\bx^\dg H)$. Throughout this work we will assume without loss of generality that $c_\bze = 0$ (i.e. $H$ is traceless), as changing $c_\bze$ corresponds merely to shifting the objective function by a uniform constant which in no way affects the geometric features of the landscape. $| \bx |$ is defined to be the number of non-zero elements of $\bx$, and $\supp(\bx)$ is defined to be the set of indices on which $\bx$ is non-zero.

Letting $f_V(\but)$ denote the objective function induced by a particular realization of entangling gates, which we denote by $V$, we are interested in the typical flatness of the optimization landscape $f_V (\but)$, with respect to both the random choice of entangling gates and uniformly over the parameter space. While a natural measure of the typical flatness is $\E_V \E_{\but} \| \nabla f_V(\but) \|^2$ (where $\E_V$ averages over the realization of entangling gates and $\E_\but$ averages over the parameter space), the next lemma shows that $\E_V \E_{\but} f_V(\but)^2$ is also a good measure by directly relating it to the former quantity. In the remainder of the work, we directly study the latter quantity as it is more directly accessible with our techniques.
\begin{lemma}\label{lem:Zeroth_First_Inequalities}
With $f_V(\but)$ as defined previously,
\begin{align*}
\E_V \E_\but f_V(\but)^2 &\leq \E_V \E_\but \| \nabla f_V (\but) \|^2 \\ 
&\leq 4\br{\sum_i \| A_i \|^2 }\| H \|  \sqrt{\E_V \E_\but f_V(\but)^2}.
\end{align*}
\end{lemma}
Intuitively, such a relation holds due to the fact that $f(\but)$ is periodic and has bounded first and second derivatives; as a result of these properties, the function value has a low variance over the parameter space if and only if the gradient is typically small. We defer the proof to \Cref{app:ProofOfGradientInequalities}. A result of a similar flavor was previously shown in \cite{arrasmith2021equivalence}; however, this result does not straightforwardly extend to the present setting, and our proof strategy is different than theirs. We record another lemma here, where we introduce the convenient notation $g_\bx \defeq \E_V \E_{\but} \expval{U(\but)^\dag M_\bx U(\but)}{0}^2$.
\begin{lemma}\label{lem:expansion}
\be
\E_V \E_\but f_V(\but)^2 = \sum_\bx |c_\bx|^2 g_\bx.
\ee
\end{lemma}
That is, the terms of $H$ contribute independently to the variance of $f_V(\but)$. A similar fact was previously shown for $\partial_i f(\but)$ in \cite{Uvarov_2021}. We defer the proof to \Cref{app:markov}. Together, these two lemmas imply that we may study the barren plateau phenomenon by studying the simpler quantities $g_\bx$.

\section{Monte Carlo algorithm}
In this section we give an efficient randomized algorithm for estimating $\E_V \E_{\but} f_V(\but)^2$, applicable for any architecture compatible with the setup described previously. We only describe the algorithm in this section, deferring the derivation to \Cref{app:markov}. We note, however, that the derivation of the algorithm and subsequent analytic bounds draw heavily on the tools and analysis introduced in \cite{dalzell2020random} for mapping random quantum circuits to Markov chains, which in turn drew from a line of work studying random quantum circuits via classical statistical mechanical models (initiated by \cite{Hayden_2016} and \cite{Nahum_2018}). Indeed, the random walk $\cM_{\wt{U}}$ that we define presently was first derived and studied in that work, and we use some of their notation. But while they were concerned with studying the anti-concentration of random circuits, in this work the same random walks are utilized differently to connect to barren plateaus. 

For any variational circuit $U$, we now define the random walk $\cM_{\wt{U}}$ on $n$-bit strings $\{\mtt{I},\mtt{S}\}^n$ associated with $\wt{U}$ as follows, which corresponds to the ``biased random walk'' studied in \cite{dalzell2020random}. Each site is initialized independently with label $\mtt{S}$ with probability $1/(q+1)$, and with label $\mtt{I}$ otherwise. Now, each entangling gate in the circuit $\wt{U}$ ``acts'' on a pair of sites according to the following rules. If the two sites are in the configuration $(\mtt{I}, \mtt{I})$ or $(\mtt{S}, \mtt{S})$, then the gate leaves the configuration unchanged; but if the starting configuration is $(\mtt{I}, \mtt{S})$ or $(\mtt{S}, \mtt{I})$, the gate sends the sites to the state $(\mtt{S}, \mtt{S})$ with probability $1/(q^2+1)$ and to $(\mtt{I}, \mtt{I})$ otherwise. One consequence of this definition is that $\mtt{S}^n$ and $\mtt{I}^n$ are fixed points of the random walk. Let $\Gamma = (\vG^0, \dots, \vG^m)$ denote the random trajectory associated with this process, where $\vG^0$ is the initial configuration and $\vG^t$ is the configuration after action of the $t\textsuperscript{th}$ gate. We now relate $g_\bx$ to the random walk $\cM_{\wt{U}}$. For clarity we first introduce the notation $\supp(\vG^i)$ to denote the set of sites on which $\vG^i$ carries label $\mtt{S}$.  We then have our main result, which is proven in \Cref{app:markov}.
\begin{theorem}\label{thm:biased}
\be
g_\bx = \Pr[\supp(\bx) \subseteq \supp(\vG^m)].
\ee
\end{theorem}
Combining this theorem with \Cref{lem:expansion}, we observe that $$\E_V \E_\but f_V(\but)^2 = \br{\sum_\bx |c_\bx|^2}\Pr[\supp(\mathbf{X}) \subseteq \supp(\vG^m)],$$ where $\mathbf{X}$ is a random variable satisfying $\Pr[\mathbf{X}=\bx] \propto |c_\bx|^2$. From this observation and a standard Chernoff bound, we immediately have the following randomized algorithm for estimating $\E_V \E_{\but} f_V(\but)^2$.
\begin{algorithm}[H]
  \caption{Algorithm for $\E_V \E_{\but} f_V(\but)^2$}\label{alg}
  \label{algorithm}
  \hspace*{\algorithmicindent} \textbf{Input:} variational circuit $U(\but)$; $H = \sum_\bx c_\bx M_\bx$; additive error tolerance $\ep$; failure probability tolerance $\delta$ \\
 \hspace*{\algorithmicindent} \textbf{Output:} estimate of $\E_V \E_{\but} f_V(\but)^2$ \\
 \hspace*{\algorithmicindent} \textbf{Runtime:} $O\br{m (\sum_\bx |c_\bx|^2)^2 \log(\delta^{-1}) \ep^{-2}}$
   \begin{algorithmic}[1]
   \State $N \gets \frac{1}{2}\log\br{\frac{2}{\delta}}\br{\frac{\sum_\bx |c_\bx|^2}{\ep}}^2$
   \For{$i \in [N]$}
   		\State Sample a realization $\vg^m$ of $\vG^m$ by simulating $\cM_{\wt{U}}$
   		\State Sample $\mathbf{X} \in [q]^n$ s.t. $\Pr[\mathbf{X}=\bx] \propto |c_\bx|^2$
   		\State $a_i \gets \mathbbm{1}_{\supp(\mathbf{X}) \subseteq \supp(\vg^m)}$
   \EndFor
   \State \textbf{return} $ (\sum_\bx |c_\bx|^2)N^{-1} \sum_{i=1}^N a_i$.
   \end{algorithmic}
\end{algorithm}
Here, $\mathbbm{1}_E$ is the random variable which is one if event $E$ occurs and zero otherwise. This algorithm provides an efficient method for quantifying the flatness of the objective landscape for any architecture compatible with our setup. It also provides an intuitive understanding of how the flatness of the landscape depends on the variational circuit depth, the locality of the objective Hamiltonian $H$, and the qudit dimension. Since the random walk on configurations is biased in favor of $\mtt{I}$ labels, we expect that over time, the population of $\mtt{S}$ labels will exponentially decay. Hence, we expect $\Pr[\supp(\bx) \subseteq \supp(\vG^m)]$---and therefore $\E_V \E_\but \| \nabla f_V(\but) \|^2$---to generally be exponentially small in the circuit depth and the size of the support of $\bx$. Meanwhile, the qudit dimension $q$ is directly related to the degree to which the walk is biased in favor of $\mtt{I}$ labels.


In the remainder of this paper we make these intuitions more quantitative, deriving bounds on the flatness of the landscape which follow largely from \Cref{thm:biased} as well as techniques previously developed in \cite{dalzell2020random} for studying anti-concentration. In particular, the arguments used to derive the bounds below follow a similar approach to arguments used in \cite{dalzell2020random} to analytically upper and lower bound a measure of anti-concentration for general and 1D architectures.  We report bounds on $g_\bx$, which are immediately related to $\E_V \E_\but f_V(\but)^2$ and $\E_V \E_\but \| \nabla f_V(\but) \|^2$ via \Cref{lem:Zeroth_First_Inequalities,lem:expansion}.

\section{General lower bound}
Consider the associated biased random walk $\cM_{\wt{U}}$ introduced in the previous section. From \Cref{thm:biased}, we see that bounding $g_\bx$ is equivalent to bounding $\Pr[\supp(\bx) \subseteq \supp(\vG^m)]$. In turn, this is lower bounded by (1) the probability that $\supp(\bx) \subseteq \supp(\Gamma^0)$ \emph{and} no spin in $\supp(\bx)$ ever flips to configuration $\mtt{I}$, as well as (2) the probability that $\vG^0 = \mtt{S}^n$, yielding
\be\label{eq:g_general_lower}
g_\bx \geq \max\cbr{\br{\frac{1}{q+1}}^{|\bx|}\br{\frac{1}{q^2+1}}^{\text{gates}(\bx)}, \br{\frac{1}{q+1}}^n},
\ee
where $\text{gates}(\bx)$ denotes the number of entangling gates acting between $\supp(\bx)$ and its complement.

To obtain a general lower bound in terms of the entangling gate circuit depth $d$, simply note that $\text{gates}(\bx) \leq d |\bx|$, due to the fact that the maximal possible number of gates which can act between sites in $\supp(\bx)$ and its complement in a single layer is $|\bx|$. Then we may concisely write a lower bound as follows.
\begin{corollary}\label{cor:general_lower}
\be
g_\bx \geq \max\br{q^{-\Theta(1) \cdot d |\bx|}, (q+1)^{-n}}.
\ee
\end{corollary}
Informally, the landscape is no flatter than an exponential in the product of the depth and locality, and its flatness is always lower bounded by $(q+1)^{-n}$.

\section{General upper bound} Here we use the random walk picture to derive a general upper bound, which we later improve in the 1D setting. We relegate the details of the calculation to \Cref{app:cor_proofs}, while describing the strategy and result here. Recalling that $g_\bx = \Pr[\supp(\bx) \subseteq \supp(\vG^m)]$, we upper bound this probability over random walks by separately considering non-zero contributions from two types of trajectories: (1) trajectories that reach the fixed point of $\mtt{S}^n$, and (2) trajectories that do not. The probability of a trajectory being of type (1) is exponentially small in $n$, intuitively due to the fact that they typically only occur if a large majority of the sites of the starting configuration $\vG^0$ are in state $\mtt{S}$, which is exponentially unlikely. The probability that a trajectory $\gamma$ is of type (2) and satisfies $\supp(\bx) \subseteq \supp(\vec{\gamma}^m)$ is exponentially small in $d$, intuitively due to the fact that, since the random walk is biased in favor of $\mtt{I}$, the population of $\mtt{S}$ labels is exponentially decaying with $d$. Precisely, we find 
\be\label{eq:g_general_upper}
g_\bx \leq \br{\frac{2q}{q+1}}^n \br{\frac{2q}{q^2+1}}^{\lfloor d/r \rfloor} q^{-|\bx|} + (q^n+1)^{-1}.
\ee

This implies the following corollaries, relevant in the large-$|\bx|$ and large-$d$ settings, respectively.
\begin{corollary}\label{cor:general_upper_global}
For $\bx$ satisfying $|\bx| \geq n/2$,
\be
g_\bx \leq q^{-\Theta(1) \cdot n}.
\ee
\end{corollary}
This corollary states that global observables always suffer from severe barren plateaus, regardless of the ansatz architecture. This effectively simplifies and generalizes one of the main results of \cite{Cerezo_2021_cost}, which analyzed particular classes of architectures and observables.

\begin{corollary}\label{cor:general_upper}
For $\bx \neq \bze$ and $d > d^* = \Theta(1) \cdot rn$,
\be
g_\bx \leq q^{-n} + q^{-\Theta(1) \cdot (d-d^*)/r  - |\bx |}.
\ee
\end{corollary}
Informally, the landscape approaches its asymptotic flatness $\sim q^{-n}$ exponentially fast in the depth once the depth exceeds $O(r\cdot n)$. Intuitively, we expect this bound to be weak, as this exponential flattening should happen well before the depth becomes $O(n)$. Indeed, it can be greatly improved for the 1D architecture studied subsequently.



\section{Improved bounds for 1D architecture}
While the bounds derived in the previous sections are quite general, we now specialize to 1D. In particular, we consider variational circuit architectures as follows. For convenience, we assume that there are an even number of qudits---numbered $\{0, 1, \dots, n-1\}$---and that boundary conditions are periodic; these two constraints are not crucial. Also assume that the entangling gates are structured as follows. At odd timesteps, an entangling gate is applied between each pair of sites $(2j, 2j+1)$ for $j \in [0, \dots, n/2 - 1]$. At even timesteps, an entangling gate is applied between each pair of sites $(2j, 2j-1 \mod n)$ for $j \in [0, \dots, n/2 - 1]$. We study $g_\bx$ as before, and additionally assume that $M_\bx$ is spatially $k$-local. 

The simplified structure of the 1D architecture permits stronger bounds to be deduced. These bounds are derived in \Cref{app:cor_proofs} and are summarized as follows.
\begin{corollary}\label{cor:1d_lower}
For 1D architecture and spatially $k$-local $M_\bx$,
\be
g_\bx \geq \max\br{q^{-\Theta(1) \cdot (d + k)}, (q+1)^{-n}}.
\ee
\end{corollary}
\begin{corollary}\label{cor:1d_upper}
For 1D architecture, $\bx \neq \bze$, spatially $k$-local $M_\bx$, and $d > d^* = \Theta(1) \cdot \log k$,
\be
g_\bx \leq q^{-(\Theta(1) \cdot (d-d^*) + |\bx |)} + q^{-n}.
\ee
\end{corollary}

Hence, in this setting we rigorously show that $g_\bx$ decays exponentially in the circuit depth and locality before saturating at $\sim q^{-n}$. These 1D bounds effectively simplify and improve some results previously obtained in \cite{Cerezo_2021_cost}, and extend them to more general regimes by permitting $k$, $d$, and $H$ to be arbitrary.

\section{Discussion}
For a model of HEAs, we have given an efficient Monte Carlo algorithm for estimating the flatness of the landscape for any architecture, and derived bounds on the typical magnitude of the gradient. In general, the algorithm provides an intuitive picture for how the landscape flatness depends on the ansatz architecture and objective observable. Here, we point out a couple of additional implications which follow straightforwardly from our results. For one, since the terms of $H$ independently contribute to $\E_V \E_\but f_V(\but)^2$ (\Cref{lem:expansion}), and we expect that the contribution of a given term to $\E_V \E_\but f_V(\but)^2$ is exponentially suppressed in its locality, we expect that local optimization approaches may have a tendency to under-weight terms of $H$ which act on more sites, since gradients may be dominated by contributions from low-weight terms. Or, more intuitively, we expect that most local minima in the optimization landscape will be mostly attributable to low-weight terms. This means that, even in architectures for which the barren plateau phenomenon is avoided, gradient descent may often be led to approximately the same local optima that they would have been for a truncated version of the Hamiltonian with terms of high locality removed. For an extreme illustration of this point, one may consider a constant-depth ansatz and an $H$ whose terms are all either $1$-local or have weight $n$. The contribution of the $1$-local terms to the objective function $f(\but)$ does not suffer from barren plateaus (from \Cref{cor:general_lower}), but the contribution from $n$-local terms experiences severe barren plateaus (from \Cref{cor:general_upper_global}), which we expect to manifest as ``narrow gorges'' \cite{Cerezo_2021_cost, arrasmith2021equivalence} in the landscape which local optimization algorithms will almost certainly not find. While this is an extreme example, a reflection on \Cref{thm:biased} reveals that the typical contribution of 2-local terms to the gradient can already be significantly smaller than that of 1-local terms, and the contribution of $k$-local terms will typically give a relative contribution exponentially smaller in $k$. This phenomenon may therefore already be relevant for the case of local Hamiltonians, and could be significant in practice for VQAs utilizing an HEA. Intuitively, this suggests the importance of using optimization methods that explore wide swaths of the parameter space, rather than a purely local approach like a naive gradient descent. However, as $k$ becomes larger, we expect it will become exponentially harder to ``find'' the significant contributions to the landscape coming from weight-$k$ terms.

We also point out that, while previous works \cite{McClean_2018, holmes2021connecting} have shown that an ansatz forming a 2-design is sufficient to imply small gradients, a straightforward implication of our work is that the landscape flatness does not generally saturate at the 2-design depth. For example, consider a 2D $\sqrt{n} \times \sqrt{n}$ array of qubits, and an ansatz whose entangling gates form a 2D random circuit with nearest-neighbor gates. This ansatz becomes an approximate 2-design at depth $O(\sqrt{n})$ \cite{harrow2018approximate}, but, assuming the objective observable is local, by \Cref{cor:general_lower} the gradient variance is still lower bounded by $\exp(-O(\sqrt{n}))$ at this depth. As the depth is increased, the gradient variance continues to shrink until it finally saturates at the much smaller value $\exp(-\Theta(n))$ at depth $\Theta(n)$ (by \Cref{cor:general_upper}).

\begin{acknowledgments}
I thank Eric Anschuetz, Alex Dalzell, Aram Harrow, Tongyang Li, and Beatrice Nash for feedback and helpful discussions related to this work.
\end{acknowledgments}

\onecolumngrid
\appendix
\section{Proof of \Cref{lem:Zeroth_First_Inequalities}}\label{app:ProofOfGradientInequalities}
Recalling that $f_V(\but)$ is component-wise periodic with period $2\pi$, we consider the Fourier series decomposition
\be
f_V(\but) = \sum_{\bn} \wt{f}_{V,\bn} e^{i \but \cdot \bn}.
\ee
Note that $\wt{f}_{V,\bn} = \wt{f}^*_{V,-\bn}$ since $f_V$ is real-valued. Since $f_V(\but)$ is smooth and periodic, we may differentiate to obtain
\be
\frac{\pd f_V}{\pd \but_j}(\but) = i \sum_{\bn} \bn_j \wt{f}_{V,\bn} e^{i \but\cdot \bn}.
\ee
We also have from Parseval's theorem that $\E_\but f_V(\but)^2 = \sum_{\bn} |\wt{f}_{V,\bn}|^2$. Another application of Parseval's theorem, and taking an expectation value over $V$, gives
\be\label{eq:partialSq}
\E_V \E_{\but} \br{\frac{\pd f_V}{\pd \but_j}}^2 = \sum_\bn \bn_j^2 \E_V \abr{\wt{f}_{V,\bn}}^2.
\ee
A second differentiation and application of Parseval's theorem yields
\be
\E_V \E_\but \br{\frac{\pd^2 f_V}{\pd \but_j^2}}^2 = \sum_\bn \bn_j^4 \E_V \abr{\wt{f}_{V,\bn}}^2.
\ee

Applying the Cauchy--Schwarz inequality to the r.h.s. of \Cref{eq:partialSq}, we have
\ba
\E_V \E_{\but} \br{\frac{\pd f_V}{\pd \but_j}}^2 &\leq \sqrt{\br{\sum_\bn \bn_j^4 \E_V |\wt{f}_{V,\bn}|^2}\br{\sum_\bn \E_V |\wt{f}_{V,\bn}|^2}} \\
&=\sqrt{\br{\E_V \E_\but \br{\frac{\pd^2 f_V}{\pd \but_j^2}}^2}\br{\E_V \E_\but f(\but)^2}} \\
&\leq 4 \| H \|  \|A_j\|^2 \sqrt{\E_V \E_\but f(\but)^2},
\ea
where in the last line we used the fact that $\frac{\pd^2 f_V}{\pd \but_j^2} \leq 4 \| H \|  \|A_j\|^2$, which can be seen by analytically differentiating the expression for $f_V(\but)$. This suffices to prove the upper bound. For the lower bound, returning again to \Cref{eq:partialSq}, we have
\be
\E_V \E_{\but} \br{\frac{\pd f_V}{\pd \but_j}}^2 \geq \sum_{\bn\ :\ \bn_j \neq 0} \E_V \abr{ \wt{f}_{V,\bn}}^2,
\ee
and so
\ba
\E_V \E_{\but} \| \nabla f_V \|^2 &\geq \sum_{\bn\ :\ \bn \neq \bze} \E_V \abr{\wt{f}_{V,\bn}}^2 \\
&= \sum_{\bn} \E_V \abr{\wt{f}_{V,\bn}}^2 - \E_V \br{\E_\but f_V(\but)}^2 \\
&=  \E_V \E_\but f_V(\but)^2 - \E_V \br{\E_\but f_V(\but)}^2,
\ea
where in the second-to-last line we used the fact that $\E_\but f_V(\but) = \wt{f}_{V, \bze}$, and the last line follows from Parseval's theorem. We now consider the quantity $\E_\but f_V(\but)$. By linearity, we have
\be
\E_\but f_V(\but) = \sum_\bx c_\bx \E_\but \expval{U_V(\but)^\dag M_\bx U_V(\but)}{0}.
\ee
Now, the term in which $\bx = \bze$ (so that $M_\bx$ is the identity) clearly contributes $c_\bze$ to the sum. We claim that all other terms vanish. To see why, note that in any other term, there must exist some $i$ for which $\bx_i \neq 0$. Now, we may expand $M_{\bx_i} = \sum_{k=0}^{q-1}\sum_{l=0}^{q-1} d_{k,l} \Sigma_1^k \Sigma_3^l$ where $d_{k,l}$ are coefficients and $d_{0,0} = 0$ (due to orthogonality of the $M_i$). Observe furthermore that
\ba
\E_{\alpha \in_U [0,2\pi]} \E_{\beta \in_U [0,2\pi]}(\Sigma_1^{q \alpha/2\pi} \Sigma_3^{q \beta/2\pi})^\dag \Sigma_1^k \Sigma_3^l (\Sigma_1^{q \alpha/2\pi} \Sigma_3^{q\beta/2\pi}) &= \E_{\alpha \in_U [0,2\pi]} \E_{\beta \in_U [0,2\pi]} e^{i (k\alpha- l\beta)} \Sigma_1^k \Sigma_3^l \\
&= \delta_{k,0}\delta_{l,0} I,
\ea
where $\alpha\in_U [0,2\pi]$ means $\alpha$ is chosen uniformly in the range $[0,2\pi]$, and we have used the commutation relations $\Sigma_3 \Sigma_1^\alpha = e^{2\pi i \alpha/q} \Sigma_1^\alpha \Sigma_3$ and $\Sigma_3^\beta \Sigma_1 = e^{2\pi i \beta/q} \Sigma_1 \Sigma_3^\beta$. Recalling that the final two parameterized gates applied to any given site $i$ is the parameterized gate sequence $\Sigma_1^{q \alpha_i/2\pi} \Sigma_3^{q \beta_i/2\pi}$, we therefore conclude that $\E_\but f_V (\but) = c_\bze$ and so
\ba
\E_V \E_\but \| \nabla f \|^2 &\geq \E_V \E_\but f_V(\but)^2 - c_\bze^2 \\
&= \E_V \E_\but f_V(\but)^2,
\ea
where in the final line we have recalled that $c_\bze = 0$. Putting the upper and lower bounds together, we have shown the desired lemma.

\section{Relating $\E_V \E_\but f_V (\but)^2$ to random walks (proofs of \Cref{lem:expansion}, \Cref{thm:biased}, and \Cref{alg})}\label{app:markov}
We will make heavy use of unitary integration formulas for calculating quantities such as $\E_U [U^{\ot n} \ot \bar{U}^{\ot n}]$ where $U$ is a Haar-random unitary. In particular, letting $U$ be a Haar-random unitary of dimension $d$ the following result is known \cite{Collins_2006}:
\be
\E_U [U^{\ot n} \ot \bar{U}^{\ot n}] = \sum_{\s, \ta \in \cS_n} \wg^n_d (\s, \ta) (I\ot P^n_d(\s)) \proj{\Phi^n_d} (I \ot P^n_d(\ta))^\dg,
\ee
where $\wg^n_d$ is the Weingarten function \cite{Collins_2003,Collins_2006}, $P_d^n(\p)$ is the permutation operator 
\be
P_d^n(\p) = \sum_{i_1, \dots, i_n \in [d]} \ket{i_{\pi^{-1}(1)}, \dots, i_{\pi^{-1}(n)}}\bra{i_1, \dots, i_n},
\ee and $\ket{\Phi_d^n}$ is $n$ copies of the non-normalized maximally entangled state:
\be
\ket{\Phi_d^n} \defeq \sum_{i_1, \dots, i_n \in [d]^n} \ket{i_1, \dots, i_n}\ket{i_1, \dots, i_n}.
\ee
We also define $\ket{\p_d^n} \defeq (I\ot P_d^n(\p)) \ket{\Phi_d^n}$. Specializing to the case in which $U$ is a 2-local Haar-random unitary acting on two qudits, $i$ and $j$, each of dimension $q$, this expression may be written
\be
\E_U [U^{\ot n} \ot \bar{U}^{\ot n}] = \sum_{\s, \ta \in \cS_n} \wg^n_{q^2} (\s, \ta) \ket{\sigma_q^n}_i\ket{\sigma_q^n}_j\bra{\ta_q^n}_i\bra{\ta_q^n}_j.
\ee
In particular, we will utilize the following equations which hold when $U$ is a unitary 2-design.
\ba
\E_U \br{U \ot \bar{U}} &= \frac{1}{q^2}\ket{I^1_q}_i\ket{I^1_q}_j\bra{I^1_q}_i\bra{I^1_q}_j \label{eq:n1integral} \\
\E_U \br{U \ot U \ot \bar{U} \ot \bar{U}} &= \frac{1}{q^4-1} \br{\ket{I^2_d}_i\ket{I^2_q}_j\bra{I^2_q}_i\bra{I^2_q}_j+\ket{S^2_q}_i\ket{S^2_q}_j\bra{S^2_q}_i\bra{S^2_q}_j} \\
&- \frac{1}{q^2(q^4-1)}\br{\ket{I^2_q}_i\ket{I^2_q}_j\bra{S^2_q}_i\bra{S^2_q}_j+\ket{S^2_q}_i\ket{S^2_q}_j\bra{I^2_q}_i\bra{I^2_q}_j}, \label{eq:n2integral} \notag
\ea
where $S$ denotes the swap permutation (i.e. cyclic permutation on two elements), and we have used the fact that $\wg^2_{q^2}(I) = \frac{1}{q^4-1}$ and $\wg^2_{q^2}(S) = -\frac{1}{q^2(q^4-1)}$.
Note that
\ba
\braket{I^1_q}{I^1_q} &= q \\
\braket{I^2_q}{I^2_q} &= \braket{S^2_q}{S^2_q} = q^2 \\
\braket{I^2_q}{S^2_q} &= q \\
\braket{I^1_q}{0^2} &= \braket{I^2_q}{0^4} = \braket{S^2_q}{0^4} = 1 \\
\bra{I_q^1}\br{M_i \ot I} \ket{I_q^1} &= \tr(M_i) \\
\bra{I^2_q}\br{M_{i} \ot M^\dg_{j} \ot I \ot I} \ket{I^2_q} &= \tr(M_i) \tr(M_{j}^\dg) = q^2 \delta_{i, j} \delta_{i,0} \label{eq:IMMI} \\
\bra{I^2_q}\br{M_{i} \ot M^\dg_{j} \ot I \ot I} \ket{S^2_q} &= \tr(M_i M_j^\dg) = q \delta_{i,j}. \label{eq:IMMS}
\ea

We warm up by computing the first moment $\E_V \E_\but f_V(\but)$. In general, for calculations that only involve first or second moments of $f_V(\but)$ we may simplify the circuit under consideration. Namely, since the entangling gates form unitary 2-designs, and the 2-design property is preserved under application of a fixed unitary, we may ignore the parameterized gates. Note that in making this simplification we utilized the constraint on the locations of the parameterized gates with respect to the entangling gates. We then have, letting $V_1, \dots, V_m$ denote the $m$ entangling gates,
\ba
\E_V \E_\but \expval{U^\dg(\but)HU(\but)}{0^n} &= \E_V \expval{V_1^\dg \dots V_m^\dg H V_m \dots V_1}{0^n} \\
&= \E_V \sum_\bx c_\bx \bra{I^1_{q^n}}\br{M_\bx \ot I}\br{V_m \ot \bar{V}_m} \cdots \br{V_1 \ot \bar{V}_1} \br{\ket{0^n} \ot \ket{0^n}}.
\ea
Since the gates are iid, we perform the average over each separately. Upon integrating each Haar-random unitary gate using \Cref{eq:n1integral}, we find
\be 
\E_V \E_\but f(\but) =  q^{-n} \sum_{\bx} c_{\bx} \prod_{j=1}^n \bra{I_q^1}\br{M_{\bx_j} \ot I} \ket{I_q^1} =  c_\bze = 0.
\ee

Our next step is to compute the second moment of the objective function value: $\E_\but \E_V f_V(\but)^2$. We find that the relevant techniques for doing so are, in fact, very similar to the techniques used by Dalzell et al. \cite{dalzell2020random} to study anti-concentration, particularly in studying the quantity $\E_V |\expval{V}{0^n}|^4$, where $V$ is a circuit consisting of Haar-random 2-local gates. Most notably, the mappings to classical random walks derived in the remainder of this section were essentially derived previously in that work. In particular, a comparison with their work shows that the combinatorial expansions for the two quantities are identical in the bulk of the circuit, but differ in how the late-time boundary conditions are treated. Or equivalently, the same random walks are relevant in studying both quantities, but the walks are related to the two quantities in somewhat different ways. Consequently, the analysis below involved in expressing $\E_\but \E_V f_V(\but)^2$ in terms of random walks is partially a re-derivation of their work, although our presentation and some aspects of the derivation are different. Proceeding, we have
\ba
\E_V \E_\but f(\but)^2 &= \E_V \expval{V_1^\dg \cdots V_m^\dg H V_m \cdots V_1}{0^n}^2 \\
&= \E_V \sum_{\bx, \bx'} c_\bx \bar{c}_{\bx'}  \expval{V_1^\dg \cdots V_m^\dg M_{\bx} V_m \cdots V_1}{0^n}  \expval{V_1^\dg \cdots V_m^\dg M^\dg_{\bx'} V_m \cdots V_1}{0^n} \\
&= \E_V \sum_{\bx, \bx'} c_\bx \bar{c}_{\bx'} \bra{I_{q^n}^2}\br{M_\bx \ot M_{\bx'}^\dg \ot I \ot I}\br{V_m \ot V_m \ot \bar{V}_m \ot \bar{V}_m}\cdots \br{V_1 \ot V_1 \ot \bar{V}_1 \ot \bar{V}_1}\ket{0}^{\ot 4n} \\
&= \E_V \sum_{\bx} |c_\bx |^2 \bra{I_{q^n}^2}\br{M_\bx \ot M_\bx^\dg \ot I \ot I}\br{V_m \ot V_m \ot \bar{V}_m \ot \bar{V}_m}\cdots \br{V_1 \ot V_1 \ot \bar{V}_1 \ot \bar{V}_1}\ket{0}^{\ot 4n} \\
&= \sum_\bx |c_\bx |^2 g_\bx,
\ea
where the cross-terms vanished due to \Cref{eq:IMMI} and \Cref{eq:IMMS}, proving \Cref{lem:expansion}. We now describe how $\E_V  \bra{I_{q^n}^2}\br{M_\bx \ot M_\bx^\dg \ot I \ot I}\br{V_m \ot V_m \ot \bar{V}_m \ot \bar{V}_m}\cdots \br{V_1 \ot V_1 \ot \bar{V}_1 \ot \bar{V}_1}\ket{0}^{\ot 4n} =: g_\bx$ may be computed combinatorially. First, for convenience we define $W_i := \E_{V_i} \br{V_i \ot V_i \ot \bar{V}_i \ot \bar{V}_i}$, so that
\be
g_\bx = \bra{I^2_{q^n}} \br{M_\bx \ot M_\bx^\dg \ot I \ot I} W_m \cdots W_1 \ket{0^{4n}}.
\ee
We now note that the starting state $\ket{0^{4n}}$ in this expression may be replaced by $q^{-n}(q+1)^{-n} \br{\ket{I^2_q}+\ket{S^2_q}}^{\ot n}$, as a consequence of the fact that
\ba
q^{-1}(q+1)^{-1} \bra{I^2_q}\br{\ket{I^2_q}+\ket{S^2_q}} = \braket{I^2_q}{0^4} = 1, \\
q^{-1}(q+1)^{-1} \bra{S^2_q}\br{\ket{I^2_q}+\ket{S^2_q}} = \braket{S^2_q}{0^4} = 1.
\ea	
It is furthermore straightforward to compute the action of $W$ on a pair of sites, where below we are suppressing subscripts and superscripts for simplicity:
\ba
W\ket{II} &= \ket{II} \\
W\ket{SS} &= \ket{SS} \\
W\ket{SI} &= W\ket{IS} = \frac{q}{q^2+1} \br{\ket{II}+\ket{SS}}.
\ea
Given these rules, it is now clear that $g_\bx$ may be expressed as a certain partition function. More precisely, borrowing some notation from \cite{dalzell2020random} we may express $g_\bx$ as a sum over \emph{trajectories} $\gamma = (\vg^{0}, \vg^1, \dots, \vg^m)$, where $\vg^t$ is an assignment of binary labels --- either $\mtt{I}$ or $\mtt{S}$ --- to each site at timestep $t$, which corresponds to the timestep after the application of gate $t$ but before that of gate $t+1$. $\vg^0$ is the initial assignment of labels before the application of any gates. We only consider trajectories which are \emph{valid}, defined as follows. A site's label may only change when acted on by a gate. Furthermore, labels may only change according to the following rules. If a gate acts on a $(\mtt{I}, \mtt{I})$ pair or a $(\mtt{S}, \mtt{S})$ pair, the configuration of that pair is left unchanged. If a gate acts on a $(\mtt{I}, \mtt{S})$ pair or a $(\mtt{S}, \mtt{I})$ pair, the subsequent configuration of this pair may be either $(\mtt{I}, \mtt{I})$ pair or a $(\mtt{S}, \mtt{S})$.

We furthermore associate a weight with any valid trajectory $\g$ as follows. Whenever a gate acts on a pair of sites with differing labels, a multiplicative factor of $\frac{q}{q^2+1}$ is incurred. There is also a ``boundary'' factor associated with an interaction between the final configuration of the trajectory and $\bx$. Defining $\supp(\bx)$ to be the set of sites $i$ for which $\bx_i \neq 0$, $\supp(\vg^m)$ to be the set of sites $i$ for which $\vg^m_i = \mtt{S}$, and $| \vg^m |$ to be the number of sites with label $\mtt{S}$ in the final configuration, $\vg^m$ is associated with a factor of $\mathbbm{1}_{\supp(\bx) \subseteq \supp(\vg^m)}q^{2n-|\vg^m|}$. Finally, there is a constant factor of $q^{-n}(q+1)^{-n}$ associated with the $t=0$ boundary. Putting these together, the weight of a valid trajectory $\g$ is
\be\label{eq:weight}
\weight_\bx (\g) = \br{\frac{q}{q+1}}^n q^{-|\vg^m|} \br{\frac{q}{q^2+1}}^{\flips(\g)} \mathbbm{1}_{\supp(\bx) \subseteq \supp(\vg^m)},
\ee
where $\flips(\g)$ denotes the number of times the configuration of a site flips over the course of the trajectory $\g$. We finally have the desired expression
\be
g_\bx = \sum_\g \weight_\bx (\g).
\ee
\subsection{Expressing $g_\bx$ as a biased random walk}
We now relate $g_\bx$ to the biased random walk $\cM_{\wt{U}}$ over configurations $\{\mtt{I},\mtt{S}\}^n$, which is exactly the same ``biased random walk'' studied in \cite{dalzell2020random}. Namely, let $P_b$ denote a distribution over trajectories $\gamma$ as follows. We will use $\vG^t$ to denote a random configuration of sites at timestep $t$, and $\Gamma$ to denote a random trajectory. Each site of the starting configuration $\vG^0$ is independently chosen to be $\mtt{S}$ with probability $1/(q+1)$ and $\mtt{I}$ with probability $q/(q+1)$. Subsequent configurations $\vG^1, \dots, \vG^m$ are distributed as follows. Whenever a gate acts on a pair of sites in the same configuration, their configuration is preserved with probability $1$. Whenever a gate acts on a pair of sites $(\mtt{I}, \mtt{S})$ in differing configurations, with probability $1/(q^2+1)$ the $\mtt{I}$ label is flipped to $\mtt{S}$, and with probability $q^2/(q^2+1)$ the $\mtt{S}$ label is flipped to $\mtt{I}$. It then may be verified through straightforward calculation that
\be
\weight_\bx(\g) = P_b(\gamma)\mathbbm{1}_{\supp(\bx) \subseteq \supp(\vg^m)},
\ee
so that, as desired,
\be
g_\bx = \Pr_{\Gamma \sim P_b} \sbr{\supp(\bx) \subseteq \supp(\vG^m)},
\ee
proving \Cref{thm:biased}.

We see now that there is a simple Monte Carlo algorithm for producing an unbiased estimate of $\E_V \E_\but f_V(\but)^2$. First, sample $\mathbf{X} \in [q]^n$ according to the distribution $\Pr[\mathbf{X} = \bx] = \frac{|c_\bx|^2}{\sum_\bx |c_\bx|^2}$. Second, sample a trajectory $\Gamma$ by simulating the biased random walk described above. Finally, output $\sum_\bx |c_\bx|^2$ if  $\supp(\vec{\Gamma}^{m}) \supseteq \supp(\bx)$ and zero otherwise. This requires $O(n+m)$ elementary operations (where $m$ is the number of gates) and produces an unbiased estimator. From a standard Chernoff bound, we know that to estimate the desired quantity with additive error at most $\eps$ and with failure probability at most  $\delta$, it suffices to take
\be
N = \frac{1}{2}\br{\frac{\sum_\bx |c_\bx |^2}{\epsilon}}^2 \log \br{\frac{2}{\delta}}
\ee
samples, proving the correctness of \Cref{alg}.

\subsection{Expressing $g_\bx$ as an unbiased random walk}
It will be convenient here to also associate a different, unbiased random walk over $\{\mtt{I}, \mtt{S} \}^n$ to the quantity $g_\bx$, with distribution over trajectories $\g$ we denote by $P_u$. This unbiased walk is exactly the ``unbiased walk'' originally derived and studied in Dalzell et al. \cite{dalzell2020random}. In the unbiased random walk $P_u$, each site of the starting configuration $\vG^0$ is independently chosen to be $\mtt{I}$ or $\mtt{S}$ with probability $1/2$ each. Subsequent configurations $\vG^1, \dots, \vG^m$ are distributed as follows. Whenever a gate acts on a pair of sites in the same configuration, their configuration is preserved with probability $1$, as in the biased random walk of the previous section. But, whenever a gate acts on a pair of sites $(\mtt{I}, \mtt{S})$ in differing configurations, with probability $1/2$ the $\mtt{I}$ label is flipped to $\mtt{S}$, and with probability $1/2$ the $\mtt{S}$ label is flipped to $\mtt{I}$. It then holds that
\be
\weight_\bx(\g) = P_u(\gamma)\br{\frac{2q}{q+1}}^n \br{\frac{2q}{q^2+1}}^{\flips(\g)} q^{-|\vG^m|} \mathbbm{1}_{\supp(\bx) \subseteq \supp(\vg^m)},
\ee
so that
\ba
g_\bx &= \E_{\G \sim P_u} \sbr{\text{val}(\G) \mathbbm{1}_{\supp(\bx) \subseteq \supp(\vg^m)}},\ \text{with} \label{eq:unbiased} \\
\text{val}(\g) &\defeq \br{\frac{2q}{q+1}}^n \br{\frac{2q}{q^2+1}}^{\flips(\g)} q^{-|\vG^m|}. 
\ea

\section{Proofs of corollaries}\label{app:cor_proofs}
\subsection{\Cref{cor:general_upper}}\label{app:upper}
We now derive an upper bound on  $g_\bx$. To do so, it will be helpful to define a trajectory $\g = (\vg^0, \dots, \vg^m)$ to be \emph{fixed} if all sites of $\vg^m$ are in the same configuration ($\mtt{I}^n$ or $\mtt{S}^n$) (i.e. if it reaches a fixed point). Now we may write
\begin{align}
g_\bx &= \sum_{\text{non-fixed }\g} \weight_\bx (\g) + \sum_{\text{fixed }\g} \weight_\bx (\g) \\
&=  \Pr_{\G\sim P_u}(\G \text{ not fixed})\E_{\G \sim P_u} \sbr{\text{val}(\G)\mathbbm{1}_{\supp(\bx) \subseteq \supp(\vg^m)}\;\middle\vert\; \Gamma\text{ not fixed}} + \Pr_{\G\sim P_b}(\vG^m = \mtt{S}^n) \\
&\leq \E_{\G \sim P_u} \sbr{\text{val}(\G)\mathbbm{1}_{\supp(\bx) \subseteq \supp(\vg^m)}\;\middle\vert\; \Gamma\text{ not fixed}} + \Pr_{\G\sim P_b}(\vG^m = \mtt{S}^n).
\end{align}

We bound each of these terms separately. For the latter term, we use the following result from Dalzell et al. \cite{dalzell2020random}.
\begin{lemma}[Follows directly from Corollary 2 of \cite{dalzell2020random}]\label{lem:dalzell}
Conditioned on $\vG^0$ having exactly $x$ sites in the $\mtt{S}$ configuration and $n-x$ sites in the $\mtt{I}$ configuration, the probability that the final configuration $\vG^m$ of the biased walk $P_b$ is $\mtt{S}^n$ is upper bounded by $\frac{q^{-2n + 2x}}{1-q^{-2n}} (1-q^{-2x})$. 
\end{lemma}

From this lemma and the definition of $P_b$, we may straightforwardly bound $\Pr_{\Gamma \sim P_b} (\vG^m = \mtt{S}^n)$.

\begin{lemma}\label{lem:prob_fixed}
\be
\Pr_{\Gamma \sim P_b} (\vG^m = \mtt{S}^n) \leq (q^n+1)^{-1}.
\ee
\end{lemma}
\begin{proof}
Recalling that each site of $\vG^0$ is chosen independently to be $\mtt{S}$ with probability $1/(q+1)$ and $\mtt{I}$ with probability $q/(q+1)$, we have
\ba
\Pr_{\Gamma \sim P_b} (\vG^m = \mtt{S}^n) &\leq \sum_{h = 0}^n \frac{q^{n-h}}{(q+1)^n}\frac{q^{-2n+2h}}{1-q^{-2n}}(1-q^{-2h}) \\
&= \frac{q^n}{(q+1)^n (q^{2n}-1)} \sum_{h=0}^n \binom{n}{h}(q^h-q^{-h}) \\
&= \frac{q^n\sbr{\br{q+1}^n - (q^{-1}+1)^n}}{(q+1)^n (q^{2n}-1)} \\
&= (q^n+1)^{-1},
\ea
where in the first line we used \Cref{lem:dalzell}.
\end{proof}

We next bound $\E_{\G \sim P_u} \sbr{\text{val}(\G)\;\middle\vert\; \Gamma\text{ not fixed}}$. For an $r$-regularly connected architecture, at least one gate is applied between any proper subset of sites and its complement every $r$ layers. It therefore holds that for any trajectory $\g$ in an $r$-regularly connected architecture, $\flips(\g) \geq \lfloor d/r \rfloor$, from which we have
\be
\text{val}(\g)\mathbbm{1}_{\supp(\bx) \subseteq \supp(\vg^m)} \leq \br{\frac{2q}{q+1}}^n \br{\frac{2q}{q^2+1}}^{\lfloor d/r \rfloor} q^{-|\bx|}
\ee
for any non-fixed $\g$, where we also used the fact that $|\vg^m | \geq | \bx |$ for any $\g$ satisfying $\supp(\bx) \subseteq \supp(\vg^m)$. Putting this together, we now  may write
\ba
g_\bx &\leq \E_{\G \sim P_u} \sbr{\text{val}(\G)\mathbbm{1}_{\supp(\bx) \subseteq \supp(\vg^m)}\;\middle\vert\; \Gamma\text{ not fixed}} + \Pr_{\G\sim P_b}(\vG^m = \mtt{S}^n) \\
&\leq \br{\frac{2q}{q+1}}^n \br{\frac{2q}{q^2+1}}^{\lfloor d/r \rfloor} q^{-|\bx|} + (q^n+1)^{-1},
\ea
from which the desired result follows.

\subsection{\Cref{cor:1d_lower}}\label{app:lower1D}
We now derive an improved lower bound on $g_\bx$ in the 1D setting where $M_\bx$ is spatially $k$-local. Let $\cS$ denote a set of $k$ adjacent sites in which the support of $\bx$ is contained. Now, $g_\bx$ is lower bounded by (1) the probability that $\cS \subseteq \supp(\Gamma^0)$ \emph{and} no spin in $\cS$ ever flips to configuration $\mtt{I}$, as well as (2) the probability that $\vG^0 = \mtt{S}^n$, yielding
\be\label{eq:g_general_lower}
g_\bx \geq \max\cbr{\br{\frac{1}{q+1}}^{k}\br{\frac{1}{q^2+1}}^{\text{gates}(\cS)}, \br{\frac{1}{q+1}}^n},
\ee
where $\text{gates}(\cS)$ denotes the number of entangling gates acting between $\cS$ and its complement. Noting that $\text{gates}(\cS) \leq d+1$, we have
\be
g_\bx \geq \max\cbr{\br{\frac{1}{q+1}}^{k}\br{\frac{1}{q^2+1}}^{d+1}, \br{\frac{1}{q+1}}^n},
\ee
from which \Cref{cor:1d_lower} follows.
\subsection{\Cref{cor:1d_upper}}\label{app:upper1D}
Our goal is to upper bound $g_\bx = \sum_\gamma \weight_\bx (\gamma)$ in the 1D setting. As before, $M_\bx$ is assumed to be $k$-local, with the support of $\bx$ being contained within the $k$ adjacent sites $\cS$. A first observation is that, in computing $g_\bx$, we need only consider the backwards lightcone of $\cS$. Let $n'$ denote the qudits in the backwards lightcone, and $U'(\but)$ denote the restriction of the original circuit $U(\but)$ to only those gates which act on these $n'$ qudits. Note that $n' \leq \min(n, k+2d)$. For the remainder of this section we restrict to $U'(\but)$ acting on these $n'$ sites. We then make a similar calculation to that made in \cite{dalzell2020random} to obtain anti-concentration bounds for 1D random circuits. Following their approach, we proceed by considering the simple domain wall structure of valid trajectories. Namely, when an entangling gate ``acts'' on a domain wall, the wall is pushed one unit to the left or right (depending on whether the output of the gate is $(\mtt{S}, \mtt{S})$ or $(\mtt{I}, \mtt{I})$). Two domain walls can collide and annihilate, but no new domain walls can be created. We refer the reader to Section V of \cite{dalzell2020random} for a more detailed discussion of the valid domain wall structure, and a more detailed explanation of the trajectory decomposition we describe presently.

Let $\cC$ denote the set of \emph{conserved trajectories}, defined as the set of trajectories for which the number of domain walls is conserved throughout the course of the trajectory; equivalently, $\g$ is a conserved trajectory if $\vg^0$ has the same number of domain walls as $\vg^m$. We then have the following key observation. For any valid trajectory $\g$, we may uniquely associate a pair of valid trajectories $(c_\g, f_\g)$ where $c_\g \in \cC$ is conserved and $f_\g$ is fixed, and where $c_\g$ ($f_\g$) corresponds to the conserved (non-conserved) domain walls of $\g$. More precisely, $c_\g$ is the conserved trajectory whose domain wall trajectories matches the trajectories of the conserved domain walls of $\g$ (i.e. those domain walls which are not annihilated over the course of the trajectory) and which satisfies $\vec{c}_\g^m = \vg^m$. $f_\g$ is the trajectory whose domain wall structure matches that of the non-conserved domain walls of $\g$ and has fixed point $\mtt{I}^{n'}$. Note that the mapping $\g \rightarrow (c_\g, f_\g)$ is injective. The weight of a valid trajectory $\g$ then decomposes as follows:
\be
\weight_\bx(\g) = \br{\frac{q+1}{q}}^{n'} \weight_\bx(c_\g) \weight_\bze(f_\g).
\ee
Intuitively such a decomposition is possible because, apart from the boundary contribution associated with $\bx$, the weight of a trajectory is solely determined by its domain wall structure, with domain walls contributing multiplicatively to the weight according to their length. We now proceed with bounding $g_\bx$, introducing the notation $\cC_{\geq 1}$ to denote the set of conserved trajectories with at least one domain wall.

\ba
g_\bx &= \sum_{\text{fixed }\g} \weight_\bx (\g) + \sum_{\text{non-fixed }\g} \weight_\bx (\g) \\ 
&\leq (q^{n'}+1)^{-1} + \sum_{\text{non-fixed }\g} \weight_\bx (\g) \\
&\leq (q^{n'}+1)^{-1} + \br{\frac{q+1}{q}}^{n'} \br{\sum_{\g \in \cC_{\geq 1}} \weight_\bx (\gamma)}\br{\sum_{\text{fixed }\gamma} \weight_\bze (\gamma)} \\
&\leq (q^{n'}+1)^{-1}  + \br{\frac{q+1}{q}}^{n'} \sum_{\g \in \cC_{\geq 1}} \weight_\bx (\gamma) \\
&\leq  (q^{n'}+1)^{-1}  +  q^{-|\bx|} \sum_{l=1}^{n'} \binom{n'}{l} \br{\frac{2q}{q^2+1}}^{l(d-1)}  \\
&=  (q^{n'}+1)^{-1}  +  q^{-|\bx|} \br{\frac{2q}{q^2+1}}^{d-1} \sum_{l=0}^{n' - 1} \binom{n'}{l+1} \br{\frac{2q}{q^2+1}}^{l(d-1)} \\
&\leq  (q^{n'}+1)^{-1}  +  q^{-|\bx|} \br{\frac{2q}{q^2+1}}^{d-1} n' \sum_{l=0}^{n'} \binom{n'}{l} \br{\frac{2q}{q^2+1}}^{l(d-1)} \\
&=  (q^{n'}+1)^{-1}  + q^{-|\bx|} \br{\frac{2q}{q^2+1}}^{d-1} n' \br{1 + \br{\frac{2q}{q^2+1}}^{d-1}}^{n'} \\
&\leq (q^{n'}+1)^{-1}  + q^{-|\bx|} \br{\frac{2q}{q^2+1}}^{d-1} n' e^{n' \br{\frac{2q}{q^2+1}}^{d-1}} \\
&\leq (q^{n'}+1)^{-1}  + q^{-|\bx|} \br{\frac{2q}{q^2+1}}^{d-1} n' \cdot 1.1, \text{  for $d \geq d^* = \Theta(1) (\log k)$} \\
&\leq q^{-\min(n,k+2d)} + q^{-\Theta(1) \cdot (d-d^*) - |\bx |}, \text{  for $d \geq d^* = \Theta(1) (\log k)$} \\
&\leq q^{-n} + q^{-\Theta(1) \cdot (d-d^*) - |\bx |}, \text{  for $d \geq d^* = \Theta(1) (\log k)$}.
\ea
The second line is an application of \Cref{lem:prob_fixed}. The third line is an application of the trajectory decomposition described above. In the fourth line, we have used the fact that $\sum_{\text{fixed }\gamma} \weight_\bze (\gamma) \leq g_\bze = 1$. In the fifth line, we upper bound $\sum_{\gamma \in \cC_{\geq 1}} \weight_\bx (\gamma)$ by counting domain wall trajectories, following \cite{dalzell2020random}. For such trajectories which have $l$ domain walls, there are no more than $\binom{n'}{l}$ starting configurations for the domain walls. Each time a gate acts across a domain wall, a multiplicative factor of $q/(q^2+1)$ is incurred, and there are two possible ways the domain wall can subsequently move (left or right). Hence, the contribution to the weight from all trajectories with $l$ conserved domain walls consistent with the boundary conditions imposed by $\bx$ is upper bounded by $q^{-|\bx|} \binom{n'}{l}\br{\frac{2q}{q^2+1}}^{l(d-1)}$. In the final line we used the observation that $q^{-(k+2d)} \leq q^{-\Theta(1) \cdot (d-d^*) - |\bx |}$ for sufficiently small $\Theta(1)$, from which the desired expression follows assuming the constant $\Theta(1)$ in $d \geq d^* = \Theta(1) (\log k)$ is sufficiently large.

\bibliography{main}

\begin{thebibliography}{35}%
\makeatletter
\providecommand \@ifxundefined [1]{%
 \@ifx{#1\undefined}
}%
\providecommand \@ifnum [1]{%
 \ifnum #1\expandafter \@firstoftwo
 \else \expandafter \@secondoftwo
 \fi
}%
\providecommand \@ifx [1]{%
 \ifx #1\expandafter \@firstoftwo
 \else \expandafter \@secondoftwo
 \fi
}%
\providecommand \natexlab [1]{#1}%
\providecommand \enquote  [1]{``#1''}%
\providecommand \bibnamefont  [1]{#1}%
\providecommand \bibfnamefont [1]{#1}%
\providecommand \citenamefont [1]{#1}%
\providecommand \href@noop [0]{\@secondoftwo}%
\providecommand \href [0]{\begingroup \@sanitize@url \@href}%
\providecommand \@href[1]{\@@startlink{#1}\@@href}%
\providecommand \@@href[1]{\endgroup#1\@@endlink}%
\providecommand \@sanitize@url [0]{\catcode `\\12\catcode `\$12\catcode
  `\&12\catcode `\#12\catcode `\^12\catcode `\_12\catcode `\%12\relax}%
\providecommand \@@startlink[1]{}%
\providecommand \@@endlink[0]{}%
\providecommand \url  [0]{\begingroup\@sanitize@url \@url }%
\providecommand \@url [1]{\endgroup\@href {#1}{\urlprefix }}%
\providecommand \urlprefix  [0]{URL }%
\providecommand \Eprint [0]{\href }%
\providecommand \doibase [0]{https://doi.org/}%
\providecommand \selectlanguage [0]{\@gobble}%
\providecommand \bibinfo  [0]{\@secondoftwo}%
\providecommand \bibfield  [0]{\@secondoftwo}%
\providecommand \translation [1]{[#1]}%
\providecommand \BibitemOpen [0]{}%
\providecommand \bibitemStop [0]{}%
\providecommand \bibitemNoStop [0]{.\EOS\space}%
\providecommand \EOS [0]{\spacefactor3000\relax}%
\providecommand \BibitemShut  [1]{\csname bibitem#1\endcsname}%
\let\auto@bib@innerbib\@empty
\bibitem [{\citenamefont {Dalzell}\ \emph {et~al.}(2022)\citenamefont
  {Dalzell}, \citenamefont {Hunter-Jones},\ and\ \citenamefont
  {Brand\~ao}}]{dalzell2020random}%
  \BibitemOpen
  \bibfield  {author} {\bibinfo {author} {\bibfnamefont {A.~M.}\ \bibnamefont
  {Dalzell}}, \bibinfo {author} {\bibfnamefont {N.}~\bibnamefont
  {Hunter-Jones}},\ and\ \bibinfo {author} {\bibfnamefont {F.~G. S.~L.}\
  \bibnamefont {Brand\~ao}},\ }\bibfield  {title} {\bibinfo {title} {Random
  quantum circuits anticoncentrate in log depth},\ }\href
  {https://doi.org/10.1103/PRXQuantum.3.010333} {\bibfield  {journal} {\bibinfo
   {journal} {PRX Quantum}\ }\textbf {\bibinfo {volume} {3}},\ \bibinfo {pages}
  {010333} (\bibinfo {year} {2022})}\BibitemShut {NoStop}%
\bibitem [{\citenamefont {Cerezo}\ \emph {et~al.}(2020)\citenamefont {Cerezo},
  \citenamefont {Arrasmith}, \citenamefont {Babbush}, \citenamefont {Benjamin},
  \citenamefont {Endo}, \citenamefont {Fujii}, \citenamefont {McClean},
  \citenamefont {Mitarai}, \citenamefont {Yuan}, \citenamefont {Cincio},\ and\
  \citenamefont {Coles}}]{cerezo2020variational}%
  \BibitemOpen
  \bibfield  {author} {\bibinfo {author} {\bibfnamefont {M.}~\bibnamefont
  {Cerezo}}, \bibinfo {author} {\bibfnamefont {A.}~\bibnamefont {Arrasmith}},
  \bibinfo {author} {\bibfnamefont {R.}~\bibnamefont {Babbush}}, \bibinfo
  {author} {\bibfnamefont {S.~C.}\ \bibnamefont {Benjamin}}, \bibinfo {author}
  {\bibfnamefont {S.}~\bibnamefont {Endo}}, \bibinfo {author} {\bibfnamefont
  {K.}~\bibnamefont {Fujii}}, \bibinfo {author} {\bibfnamefont {J.~R.}\
  \bibnamefont {McClean}}, \bibinfo {author} {\bibfnamefont {K.}~\bibnamefont
  {Mitarai}}, \bibinfo {author} {\bibfnamefont {X.}~\bibnamefont {Yuan}},
  \bibinfo {author} {\bibfnamefont {L.}~\bibnamefont {Cincio}},\ and\ \bibinfo
  {author} {\bibfnamefont {P.~J.}\ \bibnamefont {Coles}},\ }\href@noop {}
  {\bibinfo {title} {Variational quantum algorithms}} (\bibinfo {year}
  {2020}),\ \Eprint {https://arxiv.org/abs/2012.09265} {arXiv:2012.09265
  [quant-ph]} \BibitemShut {NoStop}%
\bibitem [{\citenamefont {Peruzzo}\ \emph {et~al.}(2014)\citenamefont
  {Peruzzo}, \citenamefont {McClean}, \citenamefont {Shadbolt}, \citenamefont
  {Yung}, \citenamefont {Zhou}, \citenamefont {Love}, \citenamefont
  {Aspuru-Guzik},\ and\ \citenamefont {O’brien}}]{peruzzo2014variational}%
  \BibitemOpen
  \bibfield  {author} {\bibinfo {author} {\bibfnamefont {A.}~\bibnamefont
  {Peruzzo}}, \bibinfo {author} {\bibfnamefont {J.}~\bibnamefont {McClean}},
  \bibinfo {author} {\bibfnamefont {P.}~\bibnamefont {Shadbolt}}, \bibinfo
  {author} {\bibfnamefont {M.-H.}\ \bibnamefont {Yung}}, \bibinfo {author}
  {\bibfnamefont {X.-Q.}\ \bibnamefont {Zhou}}, \bibinfo {author}
  {\bibfnamefont {P.~J.}\ \bibnamefont {Love}}, \bibinfo {author}
  {\bibfnamefont {A.}~\bibnamefont {Aspuru-Guzik}},\ and\ \bibinfo {author}
  {\bibfnamefont {J.~L.}\ \bibnamefont {O’brien}},\ }\bibfield  {title}
  {\bibinfo {title} {A variational eigenvalue solver on a photonic quantum
  processor},\ }\href@noop {} {\bibfield  {journal} {\bibinfo  {journal}
  {Nature Communications}\ }\textbf {\bibinfo {volume} {5}},\ \bibinfo {pages}
  {1} (\bibinfo {year} {2014})}\BibitemShut {NoStop}%
\bibitem [{\citenamefont {Farhi}\ \emph {et~al.}(2014)\citenamefont {Farhi},
  \citenamefont {Goldstone},\ and\ \citenamefont {Gutmann}}]{farhi2014quantum}%
  \BibitemOpen
  \bibfield  {author} {\bibinfo {author} {\bibfnamefont {E.}~\bibnamefont
  {Farhi}}, \bibinfo {author} {\bibfnamefont {J.}~\bibnamefont {Goldstone}},\
  and\ \bibinfo {author} {\bibfnamefont {S.}~\bibnamefont {Gutmann}},\
  }\href@noop {} {\bibinfo {title} {A quantum approximate optimization
  algorithm}} (\bibinfo {year} {2014}),\ \Eprint
  {https://arxiv.org/abs/1411.4028} {arXiv:1411.4028 [quant-ph]} \BibitemShut
  {NoStop}%
\bibitem [{\citenamefont {McClean}\ \emph {et~al.}(2018)\citenamefont
  {McClean}, \citenamefont {Boixo}, \citenamefont {Smelyanskiy}, \citenamefont
  {Babbush},\ and\ \citenamefont {Neven}}]{McClean_2018}%
  \BibitemOpen
  \bibfield  {author} {\bibinfo {author} {\bibfnamefont {J.~R.}\ \bibnamefont
  {McClean}}, \bibinfo {author} {\bibfnamefont {S.}~\bibnamefont {Boixo}},
  \bibinfo {author} {\bibfnamefont {V.~N.}\ \bibnamefont {Smelyanskiy}},
  \bibinfo {author} {\bibfnamefont {R.}~\bibnamefont {Babbush}},\ and\ \bibinfo
  {author} {\bibfnamefont {H.}~\bibnamefont {Neven}},\ }\bibfield  {title}
  {\bibinfo {title} {Barren plateaus in quantum neural network training
  landscapes},\ }\bibfield  {journal} {\bibinfo  {journal} {Nature
  Communications}\ }\textbf {\bibinfo {volume} {9}},\ \href
  {https://doi.org/10.1038/s41467-018-07090-4} {10.1038/s41467-018-07090-4}
  (\bibinfo {year} {2018})\BibitemShut {NoStop}%
\bibitem [{\citenamefont {Hochreiter}\ \emph {et~al.}(2001)\citenamefont
  {Hochreiter}, \citenamefont {Bengio}, \citenamefont {Frasconi},\ and\
  \citenamefont {Schmidhuber}}]{hochreiter_2001}%
  \BibitemOpen
  \bibfield  {author} {\bibinfo {author} {\bibfnamefont {S.}~\bibnamefont
  {Hochreiter}}, \bibinfo {author} {\bibfnamefont {Y.}~\bibnamefont {Bengio}},
  \bibinfo {author} {\bibfnamefont {P.}~\bibnamefont {Frasconi}},\ and\
  \bibinfo {author} {\bibfnamefont {J.}~\bibnamefont {Schmidhuber}},\ }\bibinfo
  {title} {Gradient flow in recurrent nets: The difficulty of learning
  long-term dependencies},\ in\ \href@noop {} {\emph {\bibinfo {booktitle} {A
  Field Guide to Dynamical Recurrent Networks}}},\ \bibinfo {editor} {edited
  by\ \bibinfo {editor} {\bibfnamefont {J.}~\bibnamefont {Kolen}}\ and\
  \bibinfo {editor} {\bibfnamefont {S.}~\bibnamefont {Kremer}}}\ (\bibinfo
  {publisher} {IEEE Press},\ \bibinfo {year} {2001})\BibitemShut {NoStop}%
\bibitem [{\citenamefont {Grant}\ \emph {et~al.}(2019)\citenamefont {Grant},
  \citenamefont {Wossnig}, \citenamefont {Ostaszewski},\ and\ \citenamefont
  {Benedetti}}]{Grant_2019}%
  \BibitemOpen
  \bibfield  {author} {\bibinfo {author} {\bibfnamefont {E.}~\bibnamefont
  {Grant}}, \bibinfo {author} {\bibfnamefont {L.}~\bibnamefont {Wossnig}},
  \bibinfo {author} {\bibfnamefont {M.}~\bibnamefont {Ostaszewski}},\ and\
  \bibinfo {author} {\bibfnamefont {M.}~\bibnamefont {Benedetti}},\ }\bibfield
  {title} {\bibinfo {title} {An initialization strategy for addressing barren
  plateaus in parametrized quantum circuits},\ }\href
  {https://doi.org/10.22331/q-2019-12-09-214} {\bibfield  {journal} {\bibinfo
  {journal} {Quantum}\ }\textbf {\bibinfo {volume} {3}},\ \bibinfo {pages}
  {214} (\bibinfo {year} {2019})}\BibitemShut {NoStop}%
\bibitem [{\citenamefont {Cerezo}\ \emph {et~al.}(2021)\citenamefont {Cerezo},
  \citenamefont {Sone}, \citenamefont {Volkoff}, \citenamefont {Cincio},\ and\
  \citenamefont {Coles}}]{Cerezo_2021_cost}%
  \BibitemOpen
  \bibfield  {author} {\bibinfo {author} {\bibfnamefont {M.}~\bibnamefont
  {Cerezo}}, \bibinfo {author} {\bibfnamefont {A.}~\bibnamefont {Sone}},
  \bibinfo {author} {\bibfnamefont {T.}~\bibnamefont {Volkoff}}, \bibinfo
  {author} {\bibfnamefont {L.}~\bibnamefont {Cincio}},\ and\ \bibinfo {author}
  {\bibfnamefont {P.~J.}\ \bibnamefont {Coles}},\ }\bibfield  {title} {\bibinfo
  {title} {Cost function dependent barren plateaus in shallow parametrized
  quantum circuits},\ }\bibfield  {journal} {\bibinfo  {journal} {Nature
  Communications}\ }\textbf {\bibinfo {volume} {12}},\ \href
  {https://doi.org/10.1038/s41467-021-21728-w} {10.1038/s41467-021-21728-w}
  (\bibinfo {year} {2021})\BibitemShut {NoStop}%
\bibitem [{\citenamefont {Holmes}\ \emph {et~al.}(2021)\citenamefont {Holmes},
  \citenamefont {Arrasmith}, \citenamefont {Yan}, \citenamefont {Coles},
  \citenamefont {Albrecht},\ and\ \citenamefont {Sornborger}}]{Holmes_2021}%
  \BibitemOpen
  \bibfield  {author} {\bibinfo {author} {\bibfnamefont {Z.}~\bibnamefont
  {Holmes}}, \bibinfo {author} {\bibfnamefont {A.}~\bibnamefont {Arrasmith}},
  \bibinfo {author} {\bibfnamefont {B.}~\bibnamefont {Yan}}, \bibinfo {author}
  {\bibfnamefont {P.~J.}\ \bibnamefont {Coles}}, \bibinfo {author}
  {\bibfnamefont {A.}~\bibnamefont {Albrecht}},\ and\ \bibinfo {author}
  {\bibfnamefont {A.~T.}\ \bibnamefont {Sornborger}},\ }\bibfield  {title}
  {\bibinfo {title} {Barren plateaus preclude learning scramblers},\ }\bibfield
   {journal} {\bibinfo  {journal} {Physical Review Letters}\ }\textbf {\bibinfo
  {volume} {126}},\ \href {https://doi.org/10.1103/physrevlett.126.190501}
  {10.1103/physrevlett.126.190501} (\bibinfo {year} {2021})\BibitemShut
  {NoStop}%
\bibitem [{\citenamefont {Uvarov}\ and\ \citenamefont
  {Biamonte}(2021)}]{Uvarov_2021}%
  \BibitemOpen
  \bibfield  {author} {\bibinfo {author} {\bibfnamefont {A.~V.}\ \bibnamefont
  {Uvarov}}\ and\ \bibinfo {author} {\bibfnamefont {J.~D.}\ \bibnamefont
  {Biamonte}},\ }\bibfield  {title} {\bibinfo {title} {On barren plateaus and
  cost function locality in variational quantum algorithms},\ }\href
  {https://doi.org/10.1088/1751-8121/abfac7} {\bibfield  {journal} {\bibinfo
  {journal} {Journal of Physics A: Mathematical and Theoretical}\ }\textbf
  {\bibinfo {volume} {54}},\ \bibinfo {pages} {245301} (\bibinfo {year}
  {2021})}\BibitemShut {NoStop}%
\bibitem [{\citenamefont {Cerezo}\ and\ \citenamefont
  {Coles}(2021)}]{Cerezo_2021_higher}%
  \BibitemOpen
  \bibfield  {author} {\bibinfo {author} {\bibfnamefont {M.}~\bibnamefont
  {Cerezo}}\ and\ \bibinfo {author} {\bibfnamefont {P.~J.}\ \bibnamefont
  {Coles}},\ }\bibfield  {title} {\bibinfo {title} {Higher order derivatives of
  quantum neural networks with barren plateaus},\ }\href
  {https://doi.org/10.1088/2058-9565/abf51a} {\bibfield  {journal} {\bibinfo
  {journal} {Quantum Science and Technology}\ }\textbf {\bibinfo {volume}
  {6}},\ \bibinfo {pages} {035006} (\bibinfo {year} {2021})}\BibitemShut
  {NoStop}%
\bibitem [{\citenamefont {Abbas}\ \emph {et~al.}(2021)\citenamefont {Abbas},
  \citenamefont {Sutter}, \citenamefont {Zoufal}, \citenamefont {Lucchi},
  \citenamefont {Figalli},\ and\ \citenamefont {Woerner}}]{Abbas_2021}%
  \BibitemOpen
  \bibfield  {author} {\bibinfo {author} {\bibfnamefont {A.}~\bibnamefont
  {Abbas}}, \bibinfo {author} {\bibfnamefont {D.}~\bibnamefont {Sutter}},
  \bibinfo {author} {\bibfnamefont {C.}~\bibnamefont {Zoufal}}, \bibinfo
  {author} {\bibfnamefont {A.}~\bibnamefont {Lucchi}}, \bibinfo {author}
  {\bibfnamefont {A.}~\bibnamefont {Figalli}},\ and\ \bibinfo {author}
  {\bibfnamefont {S.}~\bibnamefont {Woerner}},\ }\bibfield  {title} {\bibinfo
  {title} {The power of quantum neural networks},\ }\href
  {https://doi.org/10.1038/s43588-021-00084-1} {\bibfield  {journal} {\bibinfo
  {journal} {Nature Computational Science}\ }\textbf {\bibinfo {volume} {1}},\
  \bibinfo {pages} {403–409} (\bibinfo {year} {2021})}\BibitemShut {NoStop}%
\bibitem [{\citenamefont {Volkoff}\ and\ \citenamefont
  {Coles}(2021)}]{Volkoff_2021}%
  \BibitemOpen
  \bibfield  {author} {\bibinfo {author} {\bibfnamefont {T.}~\bibnamefont
  {Volkoff}}\ and\ \bibinfo {author} {\bibfnamefont {P.~J.}\ \bibnamefont
  {Coles}},\ }\bibfield  {title} {\bibinfo {title} {Large gradients via
  correlation in random parameterized quantum circuits},\ }\href
  {https://doi.org/10.1088/2058-9565/abd891} {\bibfield  {journal} {\bibinfo
  {journal} {Quantum Science and Technology}\ }\textbf {\bibinfo {volume}
  {6}},\ \bibinfo {pages} {025008} (\bibinfo {year} {2021})}\BibitemShut
  {NoStop}%
\bibitem [{\citenamefont {Skolik}\ \emph {et~al.}(2021)\citenamefont {Skolik},
  \citenamefont {McClean}, \citenamefont {Mohseni}, \citenamefont {van~der
  Smagt},\ and\ \citenamefont {Leib}}]{skolik2021layerwise}%
  \BibitemOpen
  \bibfield  {author} {\bibinfo {author} {\bibfnamefont {A.}~\bibnamefont
  {Skolik}}, \bibinfo {author} {\bibfnamefont {J.~R.}\ \bibnamefont {McClean}},
  \bibinfo {author} {\bibfnamefont {M.}~\bibnamefont {Mohseni}}, \bibinfo
  {author} {\bibfnamefont {P.}~\bibnamefont {van~der Smagt}},\ and\ \bibinfo
  {author} {\bibfnamefont {M.}~\bibnamefont {Leib}},\ }\bibfield  {title}
  {\bibinfo {title} {Layerwise learning for quantum neural networks},\ }\href
  {https://doi.org/10.1007/s42484-020-00036-4} {\bibfield  {journal} {\bibinfo
  {journal} {Quantum Machine Intelligence}\ }\textbf {\bibinfo {volume} {3}},\
  \bibinfo {pages} {5} (\bibinfo {year} {2021})}\BibitemShut {NoStop}%
\bibitem [{\citenamefont {Zhao}\ and\ \citenamefont {Gao}(2021)}]{Zhao_2021}%
  \BibitemOpen
  \bibfield  {author} {\bibinfo {author} {\bibfnamefont {C.}~\bibnamefont
  {Zhao}}\ and\ \bibinfo {author} {\bibfnamefont {X.-S.}\ \bibnamefont {Gao}},\
  }\bibfield  {title} {\bibinfo {title} {Analyzing the barren plateau
  phenomenon in training quantum neural networks with the zx-calculus},\ }\href
  {https://doi.org/10.22331/q-2021-06-04-466} {\bibfield  {journal} {\bibinfo
  {journal} {Quantum}\ }\textbf {\bibinfo {volume} {5}},\ \bibinfo {pages}
  {466} (\bibinfo {year} {2021})}\BibitemShut {NoStop}%
\bibitem [{\citenamefont {Sharma}\ \emph {et~al.}(2020)\citenamefont {Sharma},
  \citenamefont {Cerezo}, \citenamefont {Cincio},\ and\ \citenamefont
  {Coles}}]{sharma2020trainability}%
  \BibitemOpen
  \bibfield  {author} {\bibinfo {author} {\bibfnamefont {K.}~\bibnamefont
  {Sharma}}, \bibinfo {author} {\bibfnamefont {M.}~\bibnamefont {Cerezo}},
  \bibinfo {author} {\bibfnamefont {L.}~\bibnamefont {Cincio}},\ and\ \bibinfo
  {author} {\bibfnamefont {P.~J.}\ \bibnamefont {Coles}},\ }\href@noop {}
  {\bibinfo {title} {Trainability of dissipative perceptron-based quantum
  neural networks}} (\bibinfo {year} {2020}),\ \Eprint
  {https://arxiv.org/abs/2005.12458} {arXiv:2005.12458 [quant-ph]} \BibitemShut
  {NoStop}%
\bibitem [{\citenamefont {Wang}\ \emph {et~al.}(2021)\citenamefont {Wang},
  \citenamefont {Fontana}, \citenamefont {Cerezo}, \citenamefont {Sharma},
  \citenamefont {Sone}, \citenamefont {Cincio},\ and\ \citenamefont
  {Coles}}]{wang2021noiseinduced}%
  \BibitemOpen
  \bibfield  {author} {\bibinfo {author} {\bibfnamefont {S.}~\bibnamefont
  {Wang}}, \bibinfo {author} {\bibfnamefont {E.}~\bibnamefont {Fontana}},
  \bibinfo {author} {\bibfnamefont {M.}~\bibnamefont {Cerezo}}, \bibinfo
  {author} {\bibfnamefont {K.}~\bibnamefont {Sharma}}, \bibinfo {author}
  {\bibfnamefont {A.}~\bibnamefont {Sone}}, \bibinfo {author} {\bibfnamefont
  {L.}~\bibnamefont {Cincio}},\ and\ \bibinfo {author} {\bibfnamefont {P.~J.}\
  \bibnamefont {Coles}},\ }\bibfield  {title} {\bibinfo {title} {Noise-induced
  barren plateaus in variational quantum algorithms},\ }\href
  {https://doi.org/10.1038/s41467-021-27045-6} {\bibfield  {journal} {\bibinfo
  {journal} {Nature Communications}\ }\textbf {\bibinfo {volume} {12}},\
  \bibinfo {pages} {6961} (\bibinfo {year} {2021})}\BibitemShut {NoStop}%
\bibitem [{\citenamefont {Ortiz~Marrero}\ \emph {et~al.}(2021)\citenamefont
  {Ortiz~Marrero}, \citenamefont {Kieferov\'a},\ and\ \citenamefont
  {Wiebe}}]{marrero2021entanglement}%
  \BibitemOpen
  \bibfield  {author} {\bibinfo {author} {\bibfnamefont {C.}~\bibnamefont
  {Ortiz~Marrero}}, \bibinfo {author} {\bibfnamefont {M.}~\bibnamefont
  {Kieferov\'a}},\ and\ \bibinfo {author} {\bibfnamefont {N.}~\bibnamefont
  {Wiebe}},\ }\bibfield  {title} {\bibinfo {title} {Entanglement-induced barren
  plateaus},\ }\href {https://doi.org/10.1103/PRXQuantum.2.040316} {\bibfield
  {journal} {\bibinfo  {journal} {PRX Quantum}\ }\textbf {\bibinfo {volume}
  {2}},\ \bibinfo {pages} {040316} (\bibinfo {year} {2021})}\BibitemShut
  {NoStop}%
\bibitem [{\citenamefont {Zhang}\ \emph {et~al.}(2020)\citenamefont {Zhang},
  \citenamefont {Hsieh}, \citenamefont {Liu},\ and\ \citenamefont
  {Tao}}]{zhang2020trainability}%
  \BibitemOpen
  \bibfield  {author} {\bibinfo {author} {\bibfnamefont {K.}~\bibnamefont
  {Zhang}}, \bibinfo {author} {\bibfnamefont {M.-H.}\ \bibnamefont {Hsieh}},
  \bibinfo {author} {\bibfnamefont {L.}~\bibnamefont {Liu}},\ and\ \bibinfo
  {author} {\bibfnamefont {D.}~\bibnamefont {Tao}},\ }\href@noop {} {\bibinfo
  {title} {Toward trainability of quantum neural networks}} (\bibinfo {year}
  {2020}),\ \Eprint {https://arxiv.org/abs/2011.06258} {arXiv:2011.06258
  [quant-ph]} \BibitemShut {NoStop}%
\bibitem [{\citenamefont {Arrasmith}\ \emph
  {et~al.}(2021{\natexlab{a}})\citenamefont {Arrasmith}, \citenamefont
  {Cerezo}, \citenamefont {Czarnik}, \citenamefont {Cincio},\ and\
  \citenamefont {Coles}}]{arrasmith2020effect}%
  \BibitemOpen
  \bibfield  {author} {\bibinfo {author} {\bibfnamefont {A.}~\bibnamefont
  {Arrasmith}}, \bibinfo {author} {\bibfnamefont {M.}~\bibnamefont {Cerezo}},
  \bibinfo {author} {\bibfnamefont {P.}~\bibnamefont {Czarnik}}, \bibinfo
  {author} {\bibfnamefont {L.}~\bibnamefont {Cincio}},\ and\ \bibinfo {author}
  {\bibfnamefont {P.~J.}\ \bibnamefont {Coles}},\ }\bibfield  {title} {\bibinfo
  {title} {Effect of barren plateaus on gradient-free optimization},\ }\href
  {https://doi.org/10.22331/q-2021-10-05-558} {\bibfield  {journal} {\bibinfo
  {journal} {{Quantum}}\ }\textbf {\bibinfo {volume} {5}},\ \bibinfo {pages}
  {558} (\bibinfo {year} {2021}{\natexlab{a}})}\BibitemShut {NoStop}%
\bibitem [{\citenamefont {Pesah}\ \emph {et~al.}(2021)\citenamefont {Pesah},
  \citenamefont {Cerezo}, \citenamefont {Wang}, \citenamefont {Volkoff},
  \citenamefont {Sornborger},\ and\ \citenamefont {Coles}}]{pesah2020absence}%
  \BibitemOpen
  \bibfield  {author} {\bibinfo {author} {\bibfnamefont {A.}~\bibnamefont
  {Pesah}}, \bibinfo {author} {\bibfnamefont {M.}~\bibnamefont {Cerezo}},
  \bibinfo {author} {\bibfnamefont {S.}~\bibnamefont {Wang}}, \bibinfo {author}
  {\bibfnamefont {T.}~\bibnamefont {Volkoff}}, \bibinfo {author} {\bibfnamefont
  {A.~T.}\ \bibnamefont {Sornborger}},\ and\ \bibinfo {author} {\bibfnamefont
  {P.~J.}\ \bibnamefont {Coles}},\ }\bibfield  {title} {\bibinfo {title}
  {Absence of barren plateaus in quantum convolutional neural networks},\
  }\href {https://doi.org/10.1103/PhysRevX.11.041011} {\bibfield  {journal}
  {\bibinfo  {journal} {Phys. Rev. X}\ }\textbf {\bibinfo {volume} {11}},\
  \bibinfo {pages} {041011} (\bibinfo {year} {2021})}\BibitemShut {NoStop}%
\bibitem [{\citenamefont {Patti}\ \emph {et~al.}(2021)\citenamefont {Patti},
  \citenamefont {Najafi}, \citenamefont {Gao},\ and\ \citenamefont
  {Yelin}}]{patti2020entanglement}%
  \BibitemOpen
  \bibfield  {author} {\bibinfo {author} {\bibfnamefont {T.~L.}\ \bibnamefont
  {Patti}}, \bibinfo {author} {\bibfnamefont {K.}~\bibnamefont {Najafi}},
  \bibinfo {author} {\bibfnamefont {X.}~\bibnamefont {Gao}},\ and\ \bibinfo
  {author} {\bibfnamefont {S.~F.}\ \bibnamefont {Yelin}},\ }\bibfield  {title}
  {\bibinfo {title} {Entanglement devised barren plateau mitigation},\ }\href
  {https://doi.org/10.1103/PhysRevResearch.3.033090} {\bibfield  {journal}
  {\bibinfo  {journal} {Phys. Rev. Research}\ }\textbf {\bibinfo {volume}
  {3}},\ \bibinfo {pages} {033090} (\bibinfo {year} {2021})}\BibitemShut
  {NoStop}%
\bibitem [{\citenamefont {Holmes}\ \emph {et~al.}(2022)\citenamefont {Holmes},
  \citenamefont {Sharma}, \citenamefont {Cerezo},\ and\ \citenamefont
  {Coles}}]{holmes2021connecting}%
  \BibitemOpen
  \bibfield  {author} {\bibinfo {author} {\bibfnamefont {Z.}~\bibnamefont
  {Holmes}}, \bibinfo {author} {\bibfnamefont {K.}~\bibnamefont {Sharma}},
  \bibinfo {author} {\bibfnamefont {M.}~\bibnamefont {Cerezo}},\ and\ \bibinfo
  {author} {\bibfnamefont {P.~J.}\ \bibnamefont {Coles}},\ }\bibfield  {title}
  {\bibinfo {title} {Connecting ansatz expressibility to gradient magnitudes
  and barren plateaus},\ }\href {https://doi.org/10.1103/PRXQuantum.3.010313}
  {\bibfield  {journal} {\bibinfo  {journal} {PRX Quantum}\ }\textbf {\bibinfo
  {volume} {3}},\ \bibinfo {pages} {010313} (\bibinfo {year}
  {2022})}\BibitemShut {NoStop}%
\bibitem [{\citenamefont {Arrasmith}\ \emph
  {et~al.}(2021{\natexlab{b}})\citenamefont {Arrasmith}, \citenamefont
  {Holmes}, \citenamefont {Cerezo},\ and\ \citenamefont
  {Coles}}]{arrasmith2021equivalence}%
  \BibitemOpen
  \bibfield  {author} {\bibinfo {author} {\bibfnamefont {A.}~\bibnamefont
  {Arrasmith}}, \bibinfo {author} {\bibfnamefont {Z.}~\bibnamefont {Holmes}},
  \bibinfo {author} {\bibfnamefont {M.}~\bibnamefont {Cerezo}},\ and\ \bibinfo
  {author} {\bibfnamefont {P.~J.}\ \bibnamefont {Coles}},\ }\href@noop {}
  {\bibinfo {title} {Equivalence of quantum barren plateaus to cost
  concentration and narrow gorges}} (\bibinfo {year} {2021}{\natexlab{b}}),\
  \Eprint {https://arxiv.org/abs/2104.05868} {arXiv:2104.05868 [quant-ph]}
  \BibitemShut {NoStop}%
\bibitem [{\citenamefont {Haug}\ and\ \citenamefont
  {Kim}(2021)}]{haug2021optimal}%
  \BibitemOpen
  \bibfield  {author} {\bibinfo {author} {\bibfnamefont {T.}~\bibnamefont
  {Haug}}\ and\ \bibinfo {author} {\bibfnamefont {M.~S.}\ \bibnamefont {Kim}},\
  }\href@noop {} {\bibinfo {title} {Optimal training of variational quantum
  algorithms without barren plateaus}} (\bibinfo {year} {2021}),\ \Eprint
  {https://arxiv.org/abs/2104.14543} {arXiv:2104.14543 [quant-ph]} \BibitemShut
  {NoStop}%
\bibitem [{\citenamefont {Larocca}\ \emph {et~al.}(2021)\citenamefont
  {Larocca}, \citenamefont {Czarnik}, \citenamefont {Sharma}, \citenamefont
  {Muraleedharan}, \citenamefont {Coles},\ and\ \citenamefont
  {Cerezo}}]{larocca2021diagnosing}%
  \BibitemOpen
  \bibfield  {author} {\bibinfo {author} {\bibfnamefont {M.}~\bibnamefont
  {Larocca}}, \bibinfo {author} {\bibfnamefont {P.}~\bibnamefont {Czarnik}},
  \bibinfo {author} {\bibfnamefont {K.}~\bibnamefont {Sharma}}, \bibinfo
  {author} {\bibfnamefont {G.}~\bibnamefont {Muraleedharan}}, \bibinfo {author}
  {\bibfnamefont {P.~J.}\ \bibnamefont {Coles}},\ and\ \bibinfo {author}
  {\bibfnamefont {M.}~\bibnamefont {Cerezo}},\ }\href@noop {} {\bibinfo {title}
  {Diagnosing barren plateaus with tools from quantum optimal control}}
  (\bibinfo {year} {2021}),\ \Eprint {https://arxiv.org/abs/2105.14377}
  {arXiv:2105.14377 [quant-ph]} \BibitemShut {NoStop}%
\bibitem [{\citenamefont {Kim}\ and\ \citenamefont
  {Oz}(2021)}]{kim2021entanglement}%
  \BibitemOpen
  \bibfield  {author} {\bibinfo {author} {\bibfnamefont {J.}~\bibnamefont
  {Kim}}\ and\ \bibinfo {author} {\bibfnamefont {Y.}~\bibnamefont {Oz}},\
  }\href@noop {} {\bibinfo {title} {Entanglement diagnostics for efficient
  quantum computation}} (\bibinfo {year} {2021}),\ \Eprint
  {https://arxiv.org/abs/2102.12534} {arXiv:2102.12534 [quant-ph]} \BibitemShut
  {NoStop}%
\bibitem [{\citenamefont {Anschuetz}(2022)}]{anschuetz2022critical}%
  \BibitemOpen
  \bibfield  {author} {\bibinfo {author} {\bibfnamefont {E.~R.}\ \bibnamefont
  {Anschuetz}},\ }\bibfield  {title} {\bibinfo {title} {Critical points in
  quantum generative models},\ }in\ \href
  {https://openreview.net/forum?id=2f1z55GVQN} {\emph {\bibinfo {booktitle}
  {International Conference on Learning Representations}}}\ (\bibinfo {year}
  {2022})\BibitemShut {NoStop}%
\bibitem [{\citenamefont {Sack}\ \emph {et~al.}(2022)\citenamefont {Sack},
  \citenamefont {Medina}, \citenamefont {Michailidis}, \citenamefont {Kueng},\
  and\ \citenamefont {Serbyn}}]{sack2022avoiding}%
  \BibitemOpen
  \bibfield  {author} {\bibinfo {author} {\bibfnamefont {S.~H.}\ \bibnamefont
  {Sack}}, \bibinfo {author} {\bibfnamefont {R.~A.}\ \bibnamefont {Medina}},
  \bibinfo {author} {\bibfnamefont {A.~A.}\ \bibnamefont {Michailidis}},
  \bibinfo {author} {\bibfnamefont {R.}~\bibnamefont {Kueng}},\ and\ \bibinfo
  {author} {\bibfnamefont {M.}~\bibnamefont {Serbyn}},\ }\href@noop {}
  {\bibinfo {title} {Avoiding barren plateaus using classical shadows}}
  (\bibinfo {year} {2022}),\ \Eprint {https://arxiv.org/abs/2201.08194}
  {arXiv:2201.08194 [quant-ph]} \BibitemShut {NoStop}%
\bibitem [{\citenamefont {Rad}\ \emph {et~al.}(2022)\citenamefont {Rad},
  \citenamefont {Seif},\ and\ \citenamefont {Linke}}]{rad2022surviving}%
  \BibitemOpen
  \bibfield  {author} {\bibinfo {author} {\bibfnamefont {A.}~\bibnamefont
  {Rad}}, \bibinfo {author} {\bibfnamefont {A.}~\bibnamefont {Seif}},\ and\
  \bibinfo {author} {\bibfnamefont {N.~M.}\ \bibnamefont {Linke}},\ }\href@noop
  {} {\bibinfo {title} {Surviving the barren plateau in variational quantum
  circuits with bayesian learning initialization}} (\bibinfo {year} {2022}),\
  \Eprint {https://arxiv.org/abs/2203.02464} {arXiv:2203.02464 [quant-ph]}
  \BibitemShut {NoStop}%
\bibitem [{\citenamefont {Hayden}\ \emph {et~al.}(2016)\citenamefont {Hayden},
  \citenamefont {Nezami}, \citenamefont {Qi}, \citenamefont {Thomas},
  \citenamefont {Walter},\ and\ \citenamefont {Yang}}]{Hayden_2016}%
  \BibitemOpen
  \bibfield  {author} {\bibinfo {author} {\bibfnamefont {P.}~\bibnamefont
  {Hayden}}, \bibinfo {author} {\bibfnamefont {S.}~\bibnamefont {Nezami}},
  \bibinfo {author} {\bibfnamefont {X.-L.}\ \bibnamefont {Qi}}, \bibinfo
  {author} {\bibfnamefont {N.}~\bibnamefont {Thomas}}, \bibinfo {author}
  {\bibfnamefont {M.}~\bibnamefont {Walter}},\ and\ \bibinfo {author}
  {\bibfnamefont {Z.}~\bibnamefont {Yang}},\ }\bibfield  {title} {\bibinfo
  {title} {Holographic duality from random tensor networks},\ }\bibfield
  {journal} {\bibinfo  {journal} {Journal of High Energy Physics}\ }\textbf
  {\bibinfo {volume} {2016}},\ \href {https://doi.org/10.1007/jhep11(2016)009}
  {10.1007/jhep11(2016)009} (\bibinfo {year} {2016})\BibitemShut {NoStop}%
\bibitem [{\citenamefont {Nahum}\ \emph {et~al.}(2018)\citenamefont {Nahum},
  \citenamefont {Vijay},\ and\ \citenamefont {Haah}}]{Nahum_2018}%
  \BibitemOpen
  \bibfield  {author} {\bibinfo {author} {\bibfnamefont {A.}~\bibnamefont
  {Nahum}}, \bibinfo {author} {\bibfnamefont {S.}~\bibnamefont {Vijay}},\ and\
  \bibinfo {author} {\bibfnamefont {J.}~\bibnamefont {Haah}},\ }\bibfield
  {title} {\bibinfo {title} {Operator spreading in random unitary circuits},\
  }\bibfield  {journal} {\bibinfo  {journal} {Physical Review X}\ }\textbf
  {\bibinfo {volume} {8}},\ \href {https://doi.org/10.1103/physrevx.8.021014}
  {10.1103/physrevx.8.021014} (\bibinfo {year} {2018})\BibitemShut {NoStop}%
\bibitem [{\citenamefont {Harrow}\ and\ \citenamefont
  {Mehraban}(2018)}]{harrow2018approximate}%
  \BibitemOpen
  \bibfield  {author} {\bibinfo {author} {\bibfnamefont {A.}~\bibnamefont
  {Harrow}}\ and\ \bibinfo {author} {\bibfnamefont {S.}~\bibnamefont
  {Mehraban}},\ }\href@noop {} {\bibinfo {title} {Approximate unitary
  $t$-designs by short random quantum circuits using nearest-neighbor and
  long-range gates}} (\bibinfo {year} {2018}),\ \Eprint
  {https://arxiv.org/abs/1809.06957} {arXiv:1809.06957 [quant-ph]} \BibitemShut
  {NoStop}%
\bibitem [{\citenamefont {Collins}\ and\ \citenamefont
  {Śniady}(2006)}]{Collins_2006}%
  \BibitemOpen
  \bibfield  {author} {\bibinfo {author} {\bibfnamefont {B.}~\bibnamefont
  {Collins}}\ and\ \bibinfo {author} {\bibfnamefont {P.}~\bibnamefont
  {Śniady}},\ }\bibfield  {title} {\bibinfo {title} {Integration with respect
  to the haar measure on unitary, orthogonal and symplectic group},\ }\href
  {https://doi.org/10.1007/s00220-006-1554-3} {\bibfield  {journal} {\bibinfo
  {journal} {Communications in Mathematical Physics}\ }\textbf {\bibinfo
  {volume} {264}},\ \bibinfo {pages} {773–795} (\bibinfo {year}
  {2006})}\BibitemShut {NoStop}%
\bibitem [{\citenamefont {Collins}(2003)}]{Collins_2003}%
  \BibitemOpen
  \bibfield  {author} {\bibinfo {author} {\bibfnamefont {B.}~\bibnamefont
  {Collins}},\ }\bibfield  {title} {\bibinfo {title} {{Moments and cumulants of
  polynomial random variables on unitary groups, the Itzykson-Zuber integral,
  and free probability}},\ }\href {https://doi.org/10.1155/S107379280320917X}
  {\bibfield  {journal} {\bibinfo  {journal} {International Mathematics
  Research Notices}\ }\textbf {\bibinfo {volume} {2003}},\ \bibinfo {pages}
  {953} (\bibinfo {year} {2003})},\ \Eprint
  {https://arxiv.org/abs/https://academic.oup.com/imrn/article-pdf/2003/17/953/1881428/2003-17-953.pdf}
  {https://academic.oup.com/imrn/article-pdf/2003/17/953/1881428/2003-17-953.pdf}
  \BibitemShut {NoStop}%
\end{thebibliography}%

\end{document}